\documentclass{article}

\usepackage{PRIMEarxiv}

\usepackage[utf8]{inputenc} 
\usepackage[T1]{fontenc}    
\usepackage{hyperref}       
\usepackage{url}            
\usepackage{booktabs}       
\usepackage{amsfonts}       
\usepackage{nicefrac}       
\usepackage{microtype}      
\usepackage{graphicx}
\usepackage{multirow}
\usepackage[numbers]{natbib}
\usepackage{float}
\usepackage{tabularx}

\usepackage{booktabs}
\usepackage{amsmath}
\usepackage{amssymb}
\usepackage{amsthm}
\usepackage{bm}
\usepackage{subcaption}
\usepackage{tabularx}
\usepackage{thmtools}
\usepackage{thm-restate}
\usepackage{multirow}

\newtheorem{proposition}{Proposition}
\newtheorem{definition}{Definition}[section]
\newtheorem*{remark}{Remark}

\usepackage{lipsum}
  
\title{Can LLMs Rank? A Tale of Triads and Triage
}

\author{
  Gaurab Pokharel \\
  Virginia Tech\\
  Alexandria, VA\\
  \texttt{gaurab@vt.edu} \\
\And
  Shafkat Farabi \\
  Virginia Tech\\
  Alexandria, VA\\
  \texttt{mfarabi@vt.edu} \\
\And
  Patrick J. Fowler \\
  Washington University in St.\ Louis\\
  St. Louis, MO\\
  \texttt{pjfowler@wustl.edu} \\
\And  
  Sanmay Das \\
  Virginia Tech\\
  Alexandria, VA\\
  \texttt{sanmay@vt.edu} \\
}

\begin{document}
\maketitle

\begin{abstract}
From housing allocation for households experiencing homelessness to triage in emergency departments, LLMs are increasingly being considered as judges of consequential decisions that require ranking people for scarce resources. Ranking large groups simultaneously is cognitively demanding and error-prone. A natural solution, drawing on decades of social choice theory,  elicits pairwise comparisons and aggregates them into a total order. However, a fundamental question remains when LLMs serve as the pairwise judge: how can a practitioner tell, before committing to a ranking, whether the LLM's judgments are sufficiently consistent to trust the result? We discuss two different ways of identifying consistency. A classical diagnostic, the coefficient of consistency $\zeta$, originally developed to measure judge reliability by counting circular triads in tournament graphs, provides a cheap, model-free measure of intra-run consistency. Various standard measures of distance between rankings, for example Kendall's $\tau$, can measure inter-run variability. We show, in both theory and practice, that these measures are independently valuable, and advocate for using both to assess reliability of rankings. We demonstrate the practical importance of our results across two high-stakes prioritization tasks: homelessness service allocation and emergency department triage. Three different leading LLMs have considerably different performance profiles across these two axes of consistency. We provide guidelines for how practitioners could think about measuring and assessing consistency before committing to a model for ranking or prioritization.
\end{abstract}


\section{Introduction}
\label{sec:introduction}

Ranking represents a ubiquitous challenge in high-stakes societal resource allocation settings. Caseworkers must prioritize households for receiving emergency housing, and triage nurses must prioritize patients in the emergency room,  for example.  
These frontline workers -- caseworkers, triage nurses, benefits administrators -- typically apply professional judgment guided by standardized instruments. In homeless services, communities use different prioritization scores to rank households for housing interventions (for many years, they used the now-disfavored Vulnerability Index - Service Prioritization Decision Assistance Tool (VI-SPDAT), for example~\cite{orgcode2015vispdat,orgcode2015vifspdat,orgcode2015tayvispdat}). In emergency departments, the Emergency Severity Index (ESI)~\cite{wolf2023esi} structures nurse assessments into a five-level acuity rating that determines who is seen first.

Interest in using large language models (LLMs) to assist with prioritization decisions continues to grow, as LLMs demonstrate increasingly sophisticated reasoning capabilities~\cite{WAN2025101859,10.1145/3772318.3791045,williams2024use}. Interest aligns with the broader LLM-as-a-judge paradigm, gaining traction for tasks that once required trained human evaluators~\cite{chatbotarena, pairs}. Recent work has examined LLM judgments across domains, such as social-service interventions and medical acuity, finding that while models can engage with the task, their assessments are often inconsistent and unreliable~\cite{stereetLevelAI, williams2024use, benhaim2024triage, sayed2025triage}. If LLMs are to inform prioritization in such situations, how can a practitioner assess whether the resulting rankings are sufficiently consistent to trust?

To begin answering this crucial question, we first examine how LLMs are actually deployed in ranking tasks. It is well known that LLM outputs are sensitive to presentation order, list length, and prompt framing~\cite{llm_not_fair, chatbotarena, Saito_Wachi_Wataoka_Akimoto_2023}, making them fragile for ranking a large number of cases at once. An alternative, grounded in social-choice theory~\cite{condorcet1785essai, kendall1940method}, is to elicit \emph{pairwise} comparisons. This involves asking a model to repeatedly prioritize one out of two cases based on some criteria (e.g., urgency) and then aggregate the responses into a total order. Pairwise elicitation simplifies each judgment, controls context length, and allows principled aggregation methods such as Rank Centrality~\cite{negahban2017rank} or Borda Count~\cite{borda}. 

Two broad approaches have been used in prior literature to examine the reliability of such rankings. The first reviews intra-run consistency in the rankings, which can be quantified by a measure of transitivity violations~\cite{kendall1940method}.
The alternative examines some measure of stability or variance of the ranking outcome, typically when the variance is over the subset of comparisons used in aggregation~\cite{Ameli_Zhuang_Stoica_Mahoney_2024, rankfusion}. However, to our knowledge, the empirical relationship between the two methods of measuring consistency has not yet been explored. Moreover, we are unaware of research that examines whether potential consistency or inconsistency in evaluation has implications for the reliability and trustworthiness of AI applications in high-stakes prioritization. These are the gaps that this paper explores.

\subsection{Background and Motivation}
In principle, one can construct a ranking method that is always internally consistent (no transitivity violations), yet has high variance. Consider a method that randomly permutes the candidates and reports all pairwise comparisons consistent with the chosen permutation. It is also possible that a system with low intra-run consistency nevertheless produces low-variance outcomes. An example of this would be any method that always outputs a fixed ranking (say, lexicographic) on the whole set of choices but encodes an arbitrary number of triad transitivity violations into its pairwise reporting. This makes clear that, at least in theory, intra-run consistency (IC) and inter-run variance (IV) are independent axes, necessitating a $2 \times 2$
taxonomy illustrated in Table~\ref{tab:ranking-methods}.

\begin{table}[ht]
\centering
\setlength{\tabcolsep}{5pt}
\renewcommand{\arraystretch}{1.3}
\newcommand{\icstrut}{\rule[-0.55cm]{0pt}{1.3cm}}
\begin{tabularx}{\columnwidth}{|c|c|>{\centering}X|>{\centering\arraybackslash}X|}
    \cline{3-4}
    \multicolumn{2}{c|}{} & \multicolumn{2}{c|}{\textbf{Inter-run variance (IV)}} \\
    \cline{3-4}
    \multicolumn{2}{c|}{} & \textbf{Low} & \textbf{High} \\
    \hline
    \multirow{2}{*}{\rotatebox[origin=c]{90}{\shortstack{\textbf{Intra-run}\\\textbf{consistency}\\\textbf{(IC)}}}}
        & \textbf{High}\icstrut & Borda, consistent triads & Random-permutation \\
    \cline{2-4}
        & \textbf{Low}\icstrut & Fixed-cycle & Noisy method, inconsistent triads \\
    \hline
\end{tabularx}
\caption{Ranking methods by intra-run consistency (IC) and inter-run variance (IV).}
\label{tab:ranking-methods}
\end{table}

The question we ask is where in this taxonomy a combination of an LLM-based (or other) approach and real-world data place us, and what factors drive intra-run consistency and inter-run variance. To measure intra-run consistency, we turn to the coefficient of consistency $\zeta$. Introduced by~\citet{kendall1940method}, $\zeta$ aims to measure judge reliability by counting circular triads in tournament graphs. Circular triads, i.e., triples $\{a,b,c\}$ for which $a \succ b \succ c \succ a$, are the minimal unit of intransitivity. These correspond to the curl component of the Hodge decomposition of~\citet{jiang2011hodgerank}; $\zeta$ is a scalar summary of its magnitude. We use $\zeta$ rather than the full decomposition because it distills inconsistency into a singular scalar that a practitioner can compute without specifying a generative model. We measure inter-run variance using Kendall's $\tau$ across runs~\cite{kendalltau}.

\subsection{Related Work}
\paragraph{Pairwise Comparison and Rank Aggregation}

The problem of aggregating pairwise comparisons into a global ranking has a long history in social choice and statistical decision theory. It dates back at least to  \citet{condorcet1785essai}'s study of majority preferences. Classical approaches include positional voting rules, such as Borda count \cite{borda} and probabilistic models for pairwise outcomes. In particular, Thurstone's Case V model \cite{thurstone1994law} and the Bradley-Terry-Luce model \cite{bradley1952rank, luce1959individual} provide foundational frameworks in which each item is associated with a latent score and pairwise choices are generated based on score differences. More recent work has developed scalable ranking algorithms for large pairwise-comparison datasets, including spectral methods such as Rank Centrality \cite{negahban2017rank}. Rank Centrality (RC) constructs a Markov chain from observed pairwise comparisons and estimates item rankings from its stationary distribution. This makes RC a natural fit for our setting of pairwise prioritization judgments.

\paragraph{Consistency Diagnostics for Tournaments}
\citet{kendall1940method} introduced the circular triad count and the coefficient of consistency $\zeta$ as a test of judge reliability in paired-comparison experiments; subsequent work developed distributional theory and extensions \citep{David_1969, alway1962distribution, Iida_2009}. These classical diagnostics operate on tournament graphs and require no generative model. A more recent line of work takes a different approach. \citet{jiang2011hodgerank} decomposes the space of pairwise comparison data into three orthogonal components using the Hodge decomposition on graphs: a gradient component (consistent with a global ranking), a curl component (local cyclicity, the source of circular triads), and a harmonic component (global inconsistencies that cannot be resolved locally). The $L_2$ norm of the curl component quantifies the total cyclicity in the data, and in this sense $\zeta$ can be understood as a scalar summary of curl magnitude: a tournament with $\zeta$ near 1 has negligible curl, while one with $\zeta$ near 0 has maximum. \citet{Perotti2024hodgerank} study HodgeRank under noise, examining how the decomposition behaves when pairwise data are corrupted. We use $\zeta$ because it can be easily computed from the tournament adjacency matrix in closed form without solving an optimization problem and can be used by any practitioner as a summary diagnostic for inconsistency without depending on the specifics of the generative model.

\paragraph{LLM-as-a-Judge -- Reliability and Intransitivity} As interest grows in using LLMs for pairwise prioritization tasks \citep{Pokharel_2025}, a parallel literature on the LLM-as-a-judge paradigm \citep{chatbotarena} has identified systematic failure modes in LLM pairwise judgments, including sensitivity to presentation order \citep{llm_not_fair, chatbotarena}, verbosity bias \citep{Saito_Wachi_Wataoka_Akimoto_2023}, and self-preference \citep{self_preference}. Beyond these biases, recent work documents a subtler problem: LLM pairwise preferences are frequently intransitive. \citet{nonTransitivityLLMJudge} find that non-transitivity is more prevalent when the items being compared are close in quality. \citet{rankfusion} count circular triads directly in passage-ranking tournaments ($n = 100$), \citet{llmJudge} and \citet{whoToTrustLLMJury} report directed 3-cycle rates on text-evaluation tasks, and several other studies measure related constructs using $k$-subset violation ratios \citep{trustjudge}, strongly connected components \citep{elspr}, and intransitivity rates under various aggregation schemes \citep{improvLLMasJudge}. The consistent finding is that intransitivity is common, varies between models and tasks, and worsens for weaker judges. Crucially, however, all report consistency metrics and ranking accuracy in separate tables. None frames the question in terms of the two independent axes we identify here: intra-run consistency and inter-run variance. 

\paragraph{AI in High-Stakes Resource Allocation}

Many public service allocation problems reduce to ranking under scarcity: caseworkers must decide which households receive limited housing interventions, triage nurses determine which patients are seen first, and child welfare workers choose which families receive intensive services \citep{cheng2022childwelfare, kube2023fair, Persad_Wertheimer_Emanuel_2009}. Standardized instruments, such as the VI-SPDAT \citep{orgcode2015vispdat, orgcode2015vifspdat, orgcode2015tayvispdat} in homelessness services and the Emergency Severity Index \citep{wolf2023esi} in emergency departments, exist precisely to impose consistency on these rankings. However, they have well-documented limitations, including racial bias \cite{cronley2022, wilkey2019coordinated} and weak predictive validity \citep{brown2018reliability}. Policies often mandate prioritizing the most vulnerable \citep{hud2017_cpd1701}, while the outcomes on which workers are evaluated favor whoever would benefit most from the intervention \citep{kube2023fair, Persad_Wertheimer_Emanuel_2009}. Certainly, frontline workers exercise professional discretion to navigate this tension \citep{lipsky1980street, pokharel2024discretionary}. Algorithmic tools are increasingly deployed in these settings to assist or replace human judgment \citep{angelova2023algorithmic, saxena2024algorithmic, starke2022fairnessperceptions}, raising questions about whether automated rankings can be trusted when the consequences are borne by vulnerable populations.
 
Recently, \citet{stereetLevelAI} explicitly examined LLM judgments in social service triage. They found that while models can engage meaningfully with prioritization tasks, their pairwise assessments are inconsistent and unreliable. 
In this paper, rather than asking whether LLMs can replicate expert rankings, we ask a more fundamental question: are LLMs internally consistent (a basic requirement for trustworthiness) in their prioritization rankings? The diagnostic framework we develop is domain-general, but the motivation is grounded in these high-stakes settings where an unreliable ranking could have direct consequences for people's lives.

Although we study allocation settings where fairness is a central concern, this work is not specifically about group fairness: we do not measure disparities across protected groups, and our diagnostics cannot detect a judge that is biased in a perfectly consistent way. The fairness concern our results speak to is instead procedural. Consistency of treatment is a canonical criterion of procedural justice \citep{leventhal1980equity}, and this procedural lens has been extended to algorithmic decision-making \citep{grgichlaca2018beyond,lee2019procedural}. An individual whose priority is determined by a judge riddled with circular triads receives a rank that depends on arbitrary features of how those contradictions happen to be resolved. Such arbitrariness becomes morally significant precisely when a single system's decisions are applied at scale \cite{creel2022leviathan}. This harm can persist even when standard group-fairness metrics are satisfied.

\subsection{Contributions}

We study the relationship between $n$, the number of candidates; $\zeta$, the Coefficient of Consistency (measuring intra-run consistency); and $\tau$, the rank distance (measuring inter-run variability) in both theory and practice.

\begin{enumerate}
    \item We theoretically characterize the behavior of $\zeta$ as the number of candidates $n$ increases. Of particular interest, we show that if the rate at which pairwise comparisons produce transitivity violations stays constant, $\zeta$ converges rapidly to a fixed floor and can be reliably estimated from small samples.
    
    \item In a common stylized model (synthetic tournaments generated from Bradley-Terry-Luce (BTL) models~\citep{bradley1952rank, luce1959individual} with known ground-truth orderings), we empirically show that $\zeta$ increases linearly with $n$ and tracks $\tau$ in tight lockstep.
    
    \item In the same model, we demonstrate a divergence in the behaviors of $\tau$ and $\zeta$ when $n$ is held constant, but the sample size of comparisons available to the rank aggregator changes. $\tau$ increases with this sample size, while $\zeta$ is, by definition, constant in expectation.
    
    \item We show that empirical behavior on real data is markedly different as a function of $n$. Across a range of datasets that span three different populations experiencing homelessness, as well as patients in emergency rooms who need triage, we show that $\zeta$ stabilizes near its asymptotic floor as $n$ increases, while $\tau$ continues to increase.
    
    \item On these real-world datasets, we also show the differences between 3 major LLM models (LLaMA, Qwen, and DeepSeek) in terms of their performance on the $\zeta$ and $\tau$ dimensions. We show that the models differ considerably in terms of which is best on each of the two dimensions -- Llama always ``winning'' on $\zeta$ while Qwen wins on $\tau$, while DeepSeek often lands in the middle on both. 
    
    \item We discuss the implications of these results for real-world use of LLMs in prioritization and triage. Our results demonstrate clearly that our taxonomy has significant practical implications for measuring consistency and reliability of ranking. Our main recommendation is to analyze models carefully using both measures rather than relying exclusively on one of $\zeta$ or $\tau$. 
\end{enumerate}

\section{Preliminaries}
\label{sec:prelim}

Having introduced intra-run consistency and inter-run variance as independent axes of reliability in Table~\ref{tab:ranking-methods}, we now develop the formal machinery needed to measure each one. For intra-run consistency, we need a measure that is computable from tournament data alone and that does not require external validation. The circular triad count and its normalization into $\zeta$ provide exactly this. For inter-run variance, we use a standard measure of distance between rankings -- Kendall's $\tau$ \cite{kendalltau}.

\subsection{Tournaments and Circular Triads}
\label{subsec:triads}
Given $n$ items, a pairwise comparison procedure applied to every pair produces a tournament: a complete graph in which every pair of items is connected by exactly one directed edge indicating the preferred item. We represent a tournament by its adjacency matrix $\bm{A} \in \{0,1\}^{n \times n}$, where $A_{ij} = 1$ if item $i$ is preferred to item $j$ and $A_{ij} = 0$ otherwise, with $A_{ij} + A_{ji} = 1$ for all $i \neq j$. The \emph{out-degree} of item $i$ is $s_i = \sum_{j \neq i} A_{ij}$, the number of items it is preferred to.

\begin{definition}[Circular Triad]
    \label{def:triad}
    A \emph{circular triad} in a tournament $\bm{A}$ is a triple $\{i, j, k\}$ such that $A_{ij} = A_{jk} = A_{ki} = 1$ or      $A_{ik} = A_{kj} = A_{ji} = 1$.
\end{definition}

\noindent Circular triads are the minimal unit of intransitivity: three pairwise preferences that admit no coherent ordering of the triple. These arise naturally in voting data, in expert judgment, and in any setting where preferences are not perfectly transitive. Their presence indicates that the pairwise preference structure is not fully consistent with any single total order, and the more of them a tournament contains, the less cleanly that tournament can be summarized by a ranking. The total circular triad count in a tournament as defined above admits a closed-form expression \citep{kendall1940method, alway1962distribution}:
    \begin{equation}
    \label{eq:Tn}
        T_n = \binom{n}{3} - \sum_{i=1}^{n} \binom{s_i}{2},
    \end{equation}
with maximum value
    \begin{equation}
    \label{eq:Tmax}
        T_{\max} = \frac{1}{24} \cdot 
        \begin{cases}
            n(n^2 - 1) & \text{if } n \text{ is odd}, \\[6pt]
            n(n^2 - 4) & \text{if } n \text{ is even}.
        \end{cases}
    \end{equation}

\begin{definition}[Coefficient of Consistency]
    \label{def:zeta}
    The coefficient of consistency \citep{kendall1940method, alway1962distribution} is
        \begin{equation}
        \label{eq:zeta}
            \zeta = 1 - \frac{T_n}{T_{\max}} \in [0, 1].
        \end{equation}
\end{definition}

\noindent A value of $\zeta = 1$ indicates a perfectly transitive tournament (zero circular triads), while $\zeta = 0$ indicates maximum cyclicity. The intuition for why $\zeta$ is a good measure of ranking consistency is straightforward: any aggregation method must resolve circular preferences, and the more such resolutions it must make, the more the output depends on the particular method of resolution rather than on a shared signal roughly consistent with a total order. 
 
Crucially, $\zeta$ can be computed from the tournament data $\bm{A}$ alone. This is precisely what a practitioner needs: after collecting pairwise comparisons and running an aggregator, they have a ranking but no way to assess its quality. Computing $\zeta$ from the same data that produced the ranking provides a diagnostic that requires no additional comparisons and no ``ground truth'' (in situations where that exists).

\paragraph{Kendall's $\tau$ as a Measure of Inter-Run Variance}

To measure the second axis of our taxonomy, inter-run variance, we need a way to quantify how much two rankings of the same $n$ items disagree. We use Kendall's $\tau$ distance \citep{kendalltau}, which counts the number of pairwise disagreements between two rankings. Given two rankings of the same $n$ items, a pair $(i,j)$ is concordant if both rankings place $i$ and $j$ in the same relative order, and discordant otherwise. The normalized Kendall $\tau$ is
    \begin{equation*}
    \label{eq:tau}
        \tau = \frac{(\text{concordant pairs}) - (\text{discordant pairs})}{\binom{n}{2}},
    \end{equation*}
which takes on values in $[-1, 1]$, with $\tau = 1$ indicating identical rankings, $\tau = -1$ indicating reversed rankings, and $\tau = 0$ indicating no more agreement than would be expected by chance. In the context of our work, given two independent runs of the same pairwise comparison and aggregation pipeline on the same set of items, $\tau$ measures inter-run agreement. A value near 1 indicates that the pipeline produces stable rankings regardless of which particular run is observed; a value near 0 indicates that the rankings are essentially unrelated across runs.

\subsection{Size Dependence of $\zeta$}
\label{subsec:zeta_size_dependence}

The definition of $\zeta$ in Equation~\eqref{eq:zeta} combines two quantities that behave differently as the number of items $n$ changes: the circular triad count $T_n$, which depends on the comparison data, and the normalization constant $T_{\max}$, which depends only on $n$. To understand how $\zeta$ evolves with $n$, it is useful to decompose it into a data-dependent term and a purely combinatorial term. This decomposition also has a practical benefit: it yields a natural estimator of $\zeta$ from incomplete tournament data, which we will need for our experiments.

\begin{definition}[Inconsistency Rate]
    \label{def:inconsistency_rate}
    The \emph{inconsistency rate} of a tournament $\bm{A}$ on $n$ items is
        \begin{equation}
        \label{eq:r}
            r = \frac{T_n}{\binom{n}{3}},
        \end{equation}
    the fraction of all triplets that form circular triads.
    \end{definition}
\noindent Rewriting Equation~\eqref{eq:zeta} in terms of $r$ gives
    \begin{equation}
    \label{eq:zeta_decomp}
        \zeta = 1 - r \cdot \frac{\binom{n}{3}}{T_{\max}}.
    \end{equation}
This separates $\zeta$ into two factors: the inconsistency rate $r$, which captures how intransitive the pairwise comparisons are, and the scaling factor $\binom{n}{3}/T_{\max}$, which is a known constant depending only on $n$. The inconsistency rate has a clean interpretation: it is the probability that a uniformly random triplet drawn from the tournament contains a circular triad. Substituting in Equation~\eqref{eq:Tmax} and simplifying yields the scaling factor in closed form:
    \begin{equation}
    \label{eq:scaling_factor}
        \frac{\binom{n}{3}}{T_{\max}} = 4 \cdot 
        \begin{cases}
            \dfrac{(n-2)}{n+1} & \text{if } n \text{ is odd}, \\[6pt]
            \dfrac{(n-1)}{n+2} & \text{if } n \text{ is even}.
            \end{cases}
    \end{equation}

\paragraph{Estimation of $\zeta$ in sparse tournaments: } The decomposition in Equation~\eqref{eq:zeta_decomp} has an immediate practical benefit. In many settings, not every pair of items can be compared: either the number of pairs is too large, or comparisons are expensive. When each pair $\{i,j\}$ is observed independently with probability $p$, a triplet $\{i,j,k\}$ is fully observed only if all three of its constituent pairs happen to be selected, which occurs with probability $p^3$. Because the sampling mechanism is independent of the comparison outcomes, the circular triad rate among fully observed triplets is an unbiased estimate of $r$ (and consequently of $zeta$). Let $\mathcal{O} = \bigl\{\{i,j,k\} : \text{all three pairs observed}\bigr\}$ denote the set of fully observed triplets, and let $T_{\text{obs}}$ be the number of circular triads among them. Then
    \begin{equation}
    \label{eq:zeta_hat}
        \hat{\zeta} = 1 - \frac{T_{\text{obs}}}{|\mathcal{O}|} \cdot \frac{\binom{n}{3}}{T_{\max}}
    \end{equation}
is an unbiased estimator of $\zeta$. The ratio $T_{\text{obs}}/|\mathcal{O}|$ estimates the inconsistency rate $r$, and the scaling factor $\binom{n}{3}/T_{\max}$ is the same constant as before. The precision of the estimator depends on $|\mathcal{O}|$, which grows as $\binom{n}{3} \cdot p^3$; even at moderate sparsity the number of fully observed triplets is typically large enough for a reliable estimate.

The decomposition in Equation~\eqref{eq:zeta_decomp} also makes it clear how $\zeta$ depends on the number of items $n$. The scaling factor $\binom{n}{3}/T_{\max}$ is an increasing function of $n$ that converges to a finite limit, so the behavior of $\zeta$ is governed by how the inconsistency rate $r$ changes with $n$.

\begin{proposition}[Asymptotic behavior of the scaling factor]
\label{prop:scaling}
    Both cases of Equation~\eqref{eq:scaling_factor} converge monotonically to $4$ from below as $n \to \infty$.
\end{proposition}
\begin{proof}
    The odd case can be written as $4 \cdot (1 - \tfrac{2}{n})/(1 + \tfrac{1}{n})$; the $n$ terms vanish as $n \to \infty$ and the ratio approaches $4$. It is strictly less than $4$ for all finite $n$ because $n - 2 < n + 1$. The even case is analogous: $4 \cdot (1 - \tfrac{1}{n})/(1 + \tfrac{2}{n}) \to 4$.
\end{proof}
\noindent Since the scaling factor increases toward $4$, Equation~\eqref{eq:zeta_decomp} shows that the trajectory of $\zeta$ as $n$ grows is entirely determined by the trajectory of $r$. Three regimes are possible:

\paragraph{Constant $r$:} If the inconsistency rate does not depend on $n$, then $\zeta(n) = 1 - r \cdot f(n)$ where $f(n) = \binom{n}{3}/T_{\max}$ is increasing in $n$. In this case $\zeta$ decreases monotonically toward the asymptote:
    \begin{equation}
    \label{eq:zeta_limit}
        \lim_{n \to \infty} \zeta(n, r) = 1 - 4r.
    \end{equation}
Note that $\zeta$ cannot become negative: since $T_n \leq T_{\max}$, the inconsistency rate is bounded by $r \leq 1/f(n)$, which converges to $1/4$, so $r \cdot f(n) \leq 1$ and $\zeta \in [0,1]$ always. \textbf{Increasing $r$:} if $r$ grows with $n$, both factors push $\zeta$ downward; this would happen if larger comparison sets introduced systematically harder distinctions.\textbf{Decreasing $r$:} If $r$ decreases with $n$, the two forces oppose. Whether $\zeta$ falls, stays flat, or rises depends on whether $r$ decreases  slower than, at the same rate as, or faster than $f(n)$ grows.

\noindent This three-way characterization frames a central empirical question of the paper. The behavior of $\zeta$ as a function of $n$ is not determined by the combinatorics of tournaments alone; it depends on the inconsistency rate, which is a property of whatever process generates the pairwise comparisons. Different comparison processes (different LLMs, different noise models, different domains) may place the system in different regimes. In Section~\ref{sec:synthetic}, we instantiate this framework with the Bradley-Terry-Luce model, where $r$ can be derived analytically as a function of the preference strength parameter $\beta$. In Section~\ref{sec:evaluating-llms}, we measure $r$ across several LLMs and datasets and identify each combination's regime.
\section{A Synthetic Model}
\label{sec:synthetic}

Whatever the aggregator -- Rank Centrality (RC) \citep{negahban2017rank}, Bradley-Terry maximum likelihood (BT-MLE) \citep{hunter2004mm}, or simple win-rate aggregation -- a practitioner ends up with a ranking in hand and no way to assess its reliability. Theoretical guarantees, where they exist, are stated in terms of population-level quantities such as the spectral gap of a transition matrix \citep{negahban2017rank,shah2017stochastically}, which are unavailable from observed data. This section establishes the relationship between $\zeta$ and $\tau$ under controlled conditions, using a standard generative model with a known ground-truth ranking so that noise can be varied systematically.

\subsection{Generative Model}
\label{subsec:model}

We require a generative model for pairwise comparisons with three properties: (1) a known true ranking against which the aggregator's output can be measured, (2) a single tunable parameter that controls comparison noise, and (3) rank-gap-dependent difficulty, so that close-in-rank items are harder to distinguish than distant ones. The Bradley-Terry-Luce (BTL) model \citep{bradley1952rank, luce1959individual} satisfies the first and third requirements by construction: each item $i \in [n]$ has a latent score $w_i > 0$, and the probability that item $i$ defeats item $j$ is $P(i \succ j) = w_i / (w_i + w_j) = \sigma(\theta_i - \theta_j)$, where $\theta_i = \log w_i$ and $\sigma$ is the logistic function. To satisfy the second requirement, we impose a structure on the score vector that collapses the $n - 1$ free parameters into a single dimension.

\begin{definition}[Equally-Spaced BTL Model]
    \label{def:esbtl}
    Fix $n$ items indexed $1, \ldots, n$ in true ranking order. Assign log-scores $\theta_i = \beta(n - i)$ for $i = 1, \ldots, n$, where $\beta \geq 0$ is the \emph{per-position discriminability}. The pairwise comparison probability is
        \begin{equation}
        \label{eq:pij}
            P(i \succ j) = \sigma\bigl(\beta(j - i)\bigr) = \frac{1}{1 + \exp\bigl(-\beta(j - i)\bigr)}.
        \end{equation}
\end{definition}

\noindent $\beta$ governs the entire noise structure. At $\beta = 0$, all comparisons are fair coin flips; as $\beta \to \infty$, all comparisons become deterministic. The probability that two adjacent items are ranked differently than in the true ranking is $\varepsilon(\beta) = 1/(1 + e^{\beta})$, which decreases monotonically in $\beta$ and represents the hardest-case pair-level error rate. This model has direct precedent in the original RC paper. In \citet{negahban2017rank}, the authors generate synthetic BTL data with geometric scores $w_i = b^{(2i - 1 - n)/2n}$, yielding a constant log-score gap of $(\log b)/n$ between consecutive items. This is the same one-parameter family as Definition~\ref{def:esbtl} under the identification $\beta = (\log b)/n$. We prefer the $\beta$ parameterization because it is directly interpretable: two experiments with the same $\beta$ but different $n$ share the same adjacent error rate $\varepsilon(\beta)$ and local comparison difficulty.

\begin{figure}[t]
    \centering
    \includegraphics[width=0.7\columnwidth]{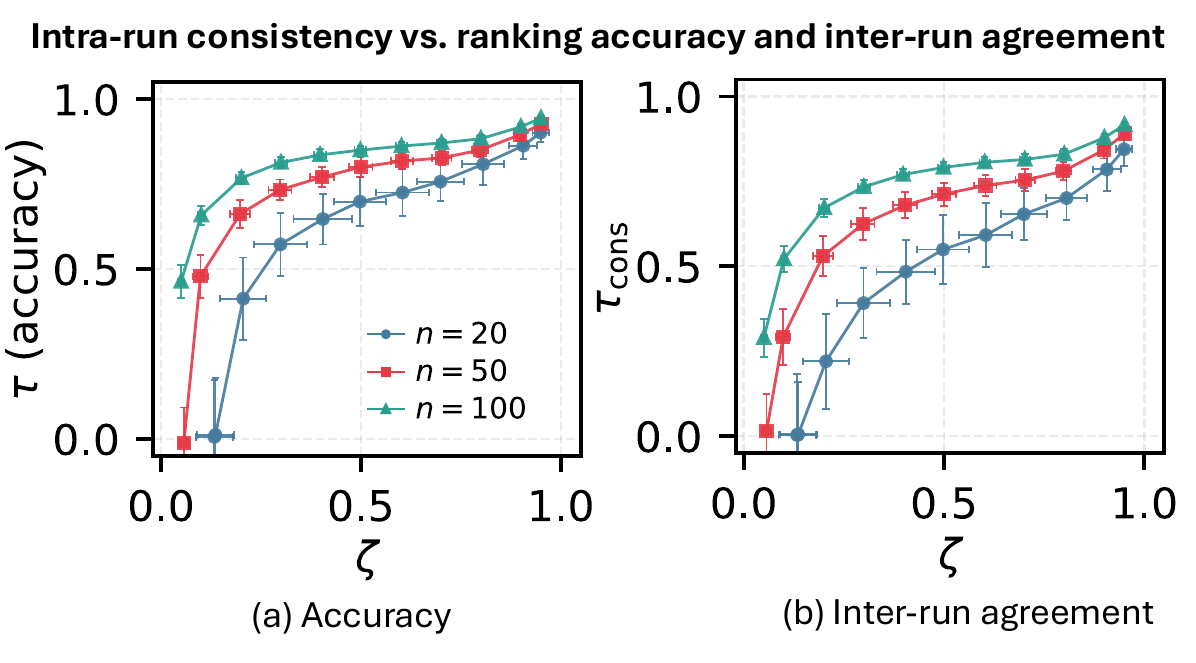}
    \caption{Intra-run consistency ($\zeta$) vs.\ the two facets of ranking quality under the synthetic BTL model, for $n \in \{20, 50, 100\}$. \textbf{(a)}~$\zeta$ vs.\ ranking accuracy (Kendall $\tau$ against the known ground-truth ordering). The relationship is monotonically increasing and largely converges across tournament sizes at high $\zeta$, with greater separation at low consistency where smaller tournaments provide less information per item. \textbf{(b)}~$\zeta$ vs.\ inter-run agreement ($\tau_{\text{cons}}$, Kendall $\tau$ between two independent RC rankings from the same generative model). The curves separate by $n$: at a given $\zeta$, larger tournaments yield higher inter-run agreement. This separation shows that intra-run consistency alone does not determine inter-run variance; the size of the comparison matrix also matters. Error bars represent $\pm 1$ standard deviation across 100 trials.}
    \label{fig:main-calibration}
\end{figure}

\subsection{Experimental Design}
\label{sec:synthetic_design}

The mathematical analysis of Section~\ref{subsec:zeta_size_dependence} showed that the behavior of $\zeta$ with $n$ is governed by the inconsistency rate $r$, which under the model defined in Definition~\ref{def:esbtl} is a function of both $\beta$ and $n$. To disentangle the roles of these two parameters, we organize the synthetic experiments along two axes.

\subsubsection{Results: Complete Tournaments}
\label{sec:synthetic_results_complete}
In the first study, we fix a grid of target $\zeta$ values and, for each pair $(\zeta^*, n)$ with $n \in \{20, 50, 100\}$, solve for the value of $\beta$ that produces expected consistency $\zeta^*$ at tournament size $n$. This inverts the natural direction of the model: rather than fixing the noise level and observing how consistent the resulting tournament is, we fix the consistency level and ask what noise level is required to achieve it. For each $(\zeta^*, n)$ pair, we generate 100 independent complete tournaments from the BTL model at the corresponding $\beta$, aggregate each via Rank Centrality (we show in Figure~\ref{fig:aggregator_comparison} in Appendix~\ref{appen:aggregator_comparison} 
that these experiments are robust to the specific choice of rank aggregator used), and measure three quantities: the observed $\zeta$ of each tournament, the accuracy of the aggregated ranking (Kendall $\tau$ against the known ground-truth ordering), and the inter-run agreement (Kendall $\tau$ between rankings produced by independent runs of the same generative process, which we denote $\tau_{\text{cons}}$).

\begin{figure}[t]
    \centering
    \includegraphics[width=0.7\columnwidth]{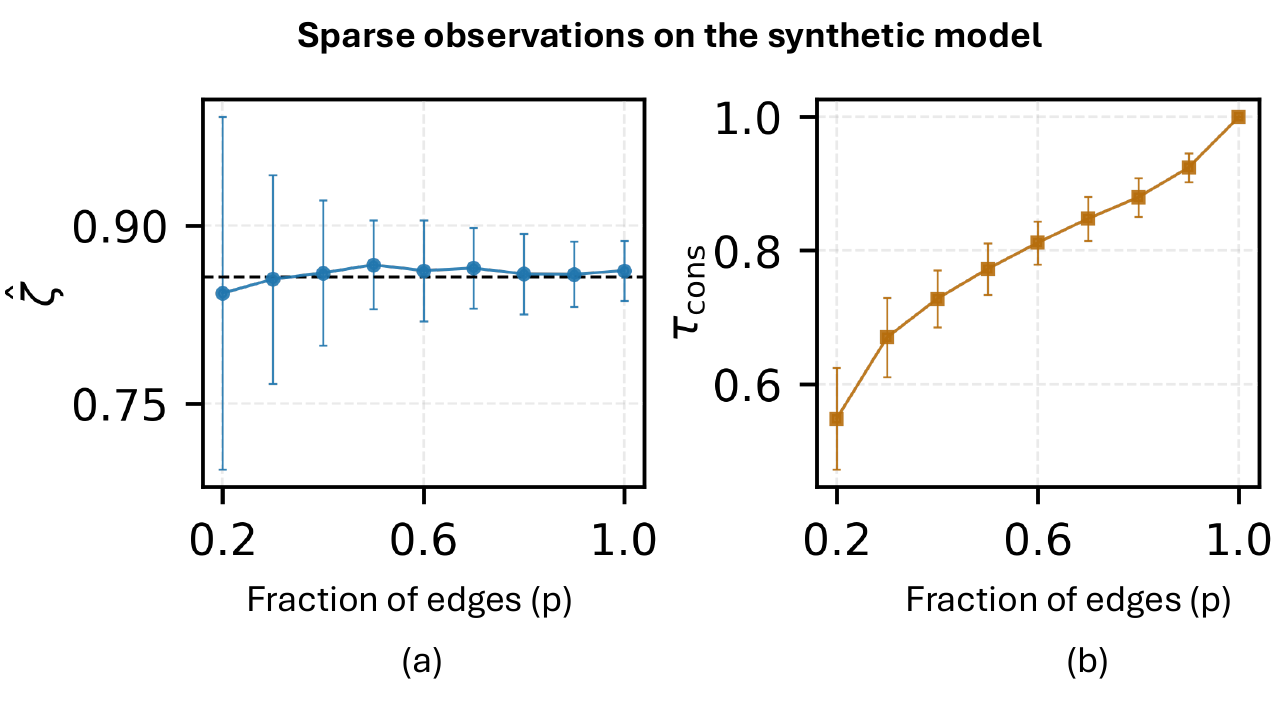}
    \caption{$n = 30$, $\beta = 0.3$, full-tournament $\zeta \approx 0.859$, marked by the dashed line. \textbf{(Left)}~The sparse estimator $\hat{\zeta}$ is unbiased across the full range of $p$; variance increases at low $p$ but the mean remains stable. \textbf{(Right)}~Inter-run agreement $\tau_{\text{cons}}$ increases monotonically with the observation fraction, from roughly $0.55$ at $p = 0.2$ to perfect agreement at $p = 1.0$. Each point averages 100 independent subsampled tournaments. The contrast illustrates that $\zeta$ measures the consistency of the comparison signal while $\tau_{\text{cons}}$ reflects both consistency and quantity.}
    \label{fig:sparse-observation}
\end{figure}

Figure 1a (left panel) shows that $\zeta$ is a strong predictor of ranking accuracy. The relationship is monotonically increasing across all three tournament sizes: higher intra-run consistency corresponds to higher agreement with the ground-truth ordering. With increasing $n$, at the same level of consistency $\zeta$, we observe a higher value of $\tau_{\mathrm{acc}}$. At high $\zeta$ ($> 0.8$), the three curves converge, indicating that when comparisons are sufficiently clean, even a modest number of items yields accurate rankings. Figure 1b (right panel) shows the relationship between $\zeta$ and inter-run agreement $\tau_{\mathrm{cons}}$. Here the curves separate similarly by $n$: at any given level of intra-run consistency, larger tournaments produce substantially more stable rankings across independent runs. At $\zeta = 0.5$, for instance, $\tau_{\text{cons}}$ ranges from roughly $0.55$ at $n = 20$ to roughly $0.79$ at $n = 100$. The explanation is that a complete tournament on $n$ items contains $\binom{n}{2}$ comparisons, so each item participates in $n - 1$ comparisons. Increasing $n$ increases the amount of data per item available to the aggregator, reducing estimation variance in the ranking even when the level of intransitivity (as captured by $\zeta$) is unchanged.
 
These results show that, in this synthetic setting, intra-run consistency and inter-run variance are correlated, but they are not redundant: two tournaments with identical $\zeta$ can have very different levels of ranking stability depending on the size of the comparison matrix.

\subsubsection{Results: Sparse Observation}
\label{sec:synthetic_results_sparse}

In the second study, we fix $n = 30$ and $\beta = 0.3$ (corresponding to $\zeta \approx 0.859$ in the full tournament) and vary the fraction of pairwise comparisons observed. For each observation probability $p \in \{0.2, 0.3, \ldots, 1.0\}$, we subsample edges uniformly at random, aggregate only the observed comparisons via RC, and measure both $\hat{\zeta}$ (using the sparse estimator from Equation~\eqref{eq:zeta_hat}) and $\tau_{\text{cons}}$ (between rankings from independent subsampled runs). This isolates the effect of comparison budget on inter-run variance while holding the underlying noise process constant.

Figure 2a confirms that the sparse estimator $\hat{\zeta}$ is unbiased across the full range of observation probabilities. At every value of $p$ from $0.2$ to $1.0$, the mean of $\hat{\zeta}$ falls on or near the dashed line marking the full-tournament value $\zeta = 0.859$. The variance of the estimator increases at low $p$, as expected given that fewer fully observed triplets are available, but there is no systematic bias. Figure 2b, on the other hand, tells a different story for inter-run agreement. $\tau_{\text{cons}}$ increases steadily from roughly $0.55$ at $p = 0.2$ to $1.0$ at $p = 1.0$. When only $20\%$ of pairs are observed, two independent subsampled runs see largely different comparisons, and their aggregated rankings diverge substantially. As the observation fraction increases, the runs share more data and the rankings converge. Taken together, Figure~\ref{fig:sparse-observation} provides a clean illustration of the distinction between our two axes. The underlying comparison process is identical across all values of $p$: the same $\beta$, the same items, the same pairwise noise. The sparse estimator correctly reports this, returning the same $\hat{\zeta}$ regardless of how many comparisons are observed. But inter-run variance depends not only on the consistency of the comparisons but on the quantity. A practitioner who observes low $\tau_{\text{cons}}$ but high $\hat{\zeta}$ can diagnose the problem as insufficient comparison budget rather than unreliable judgments, and the remedy is to elicit more comparisons rather than to distrust the comparisons themselves.

\section{Evaluating LLMs}
\label{sec:evaluating-llms}

Section~\ref{sec:synthetic} established that the two axes of our taxonomy are correlated but not redundant even under controlled conditions: $\zeta$ captures comparison consistency independently of how many comparisons are collected, while inter-run variation also depends on the comparison budget and aggregation method. The question is whether real-world pairwise comparisons generated by LLMs respect this relationship, or depart from it in informative ways.

\begin{figure}[ht]
    \centering
    \includegraphics[width=0.5\columnwidth]{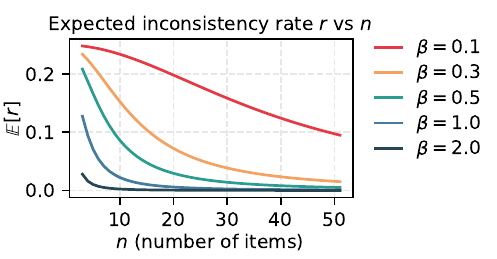}
    \caption{Expected inconsistency rate $E[r]$ as a function of $n$ for several values of $\beta$ in the equally-spaced BTL model. For every fixed $\beta > 0$, the inconsistency rate decreases with $n$, confirming Proposition~\ref{prop:zeta_scaling}.}
    \label{fig:expected-r-vs-n}
\end{figure}

\begin{restatable}{proposition}{zetascaling}
\label{prop:zeta_scaling}
    Fix $\beta > 0$. In the equally-spaced BTL model as defined in Definition~\ref{def:esbtl}, the expected number of circular triads in a complete tournament satisfies $E[T_n] = \Theta(n)$, and $T_{\max} = \Theta(n^3)$. Hence: 
    \begin{equation*}
        E[1 - \zeta] = \frac{E[T_n]}{T_{\max}} = \Theta \left(\frac{1}{n^2}\right).
    \end{equation*}
\end{restatable}
 
\begin{proof}[Proof sketch]
For a triple $i < j < k$ with rank gaps $a = j - i$ and $b = k - j$, the upset probability at gap $d$ is $q_d = \sigma(-\beta d)$, which decays exponentially in $d$. A circular triad requires either one long-edge upset or two short-edge upsets, so the cycle probability for a given triple is dominated by $e^{-\beta(a+b)}$. Summing over all triples,  the contribution from close neighbors (small rank gaps) grow linearly in $n$, while contributions from distant items are exponentially suppressed. This gives $E[T_n] = \Theta(n)$. Since $T_{\max} = \Theta(n^3)$, the ratio vanishes as $\Theta(1/n^2)$. The full proof is in Appendix~\ref{appen:zeta_scaling_proof}. 
\end{proof}

\begin{figure*}[!ht]
    \centering
    \setlength{\tabcolsep}{2pt}
    \renewcommand{\arraystretch}{0.8}
    \begin{tabular}{cccc}
        \multicolumn{4}{c}{\includegraphics[width=0.5\textwidth]{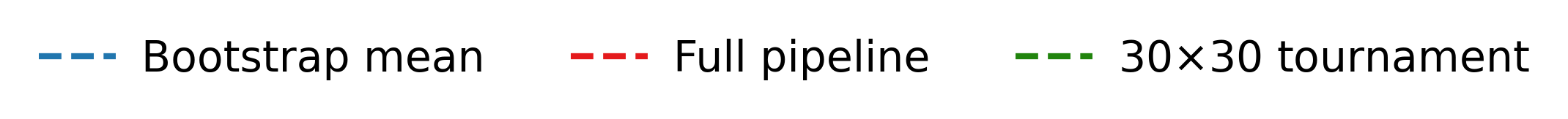}} \\[4pt]
        & \textbf{VI-SPDAT} & \textbf{VI-F-SPDAT} & \textbf{TAY-VI-SPDAT} \\[4pt]
        \rotatebox{90}{\hspace{12pt}\textbf{DeepSeek}} &
        \includegraphics[width=0.31\textwidth]{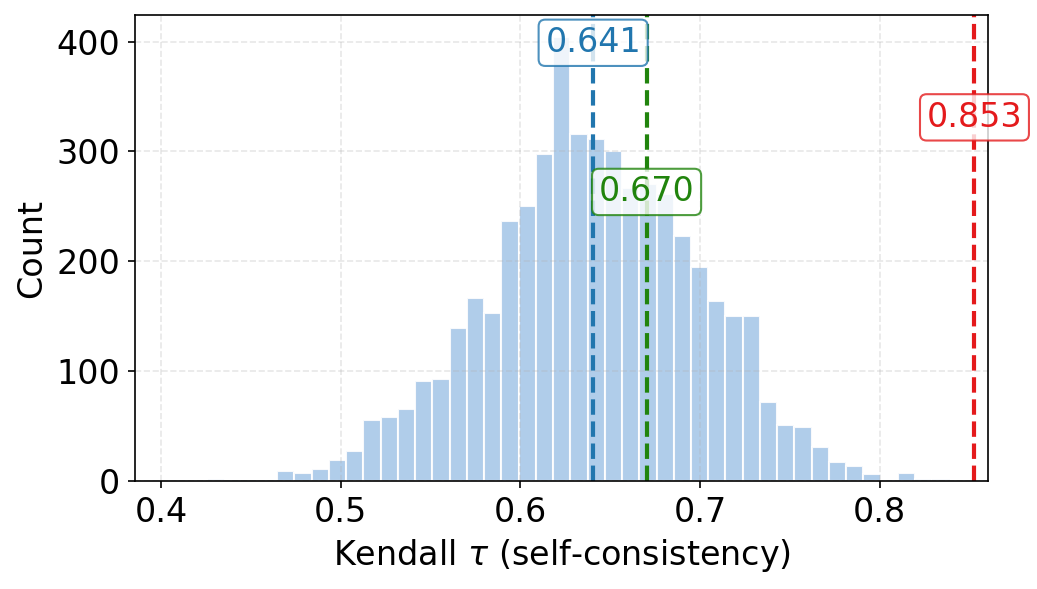} &
        \includegraphics[width=0.31\textwidth]{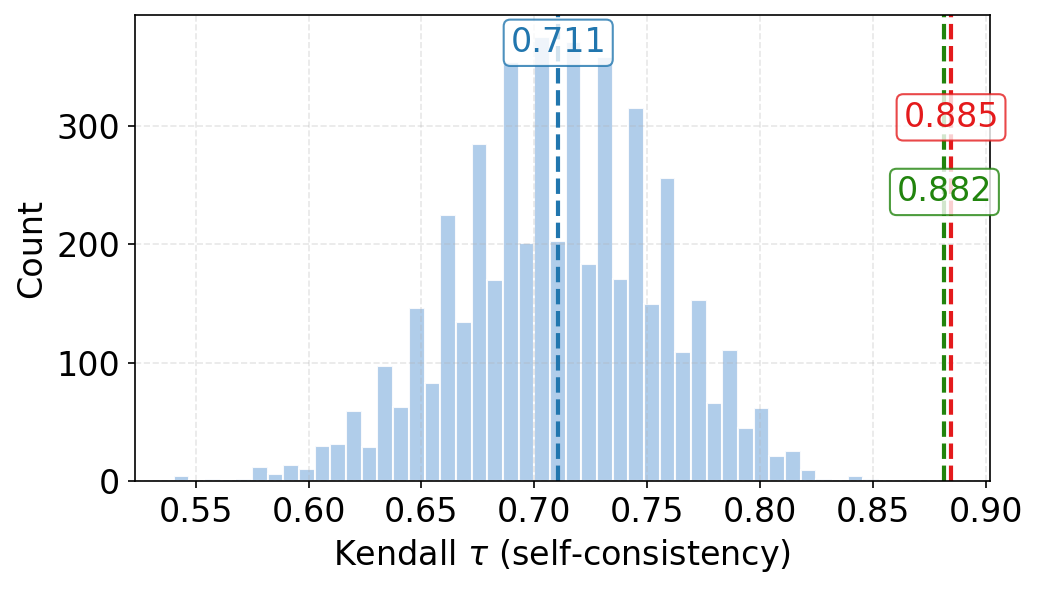} &
        \includegraphics[width=0.31\textwidth]{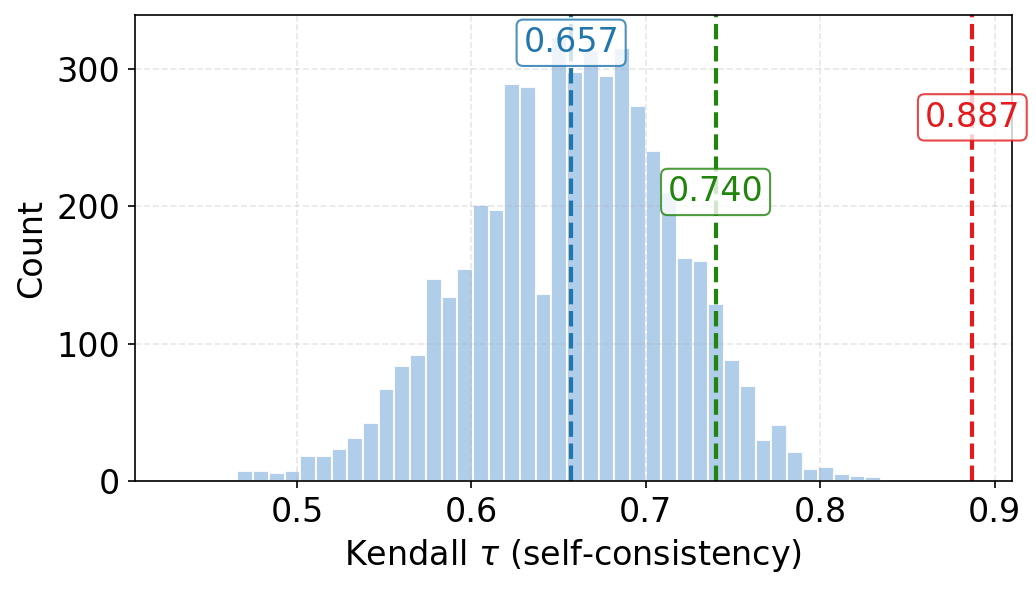} \\[2pt]
        \rotatebox{90}{\hspace{16pt}\textbf{LLaMA}} &
        \includegraphics[width=0.31\textwidth]{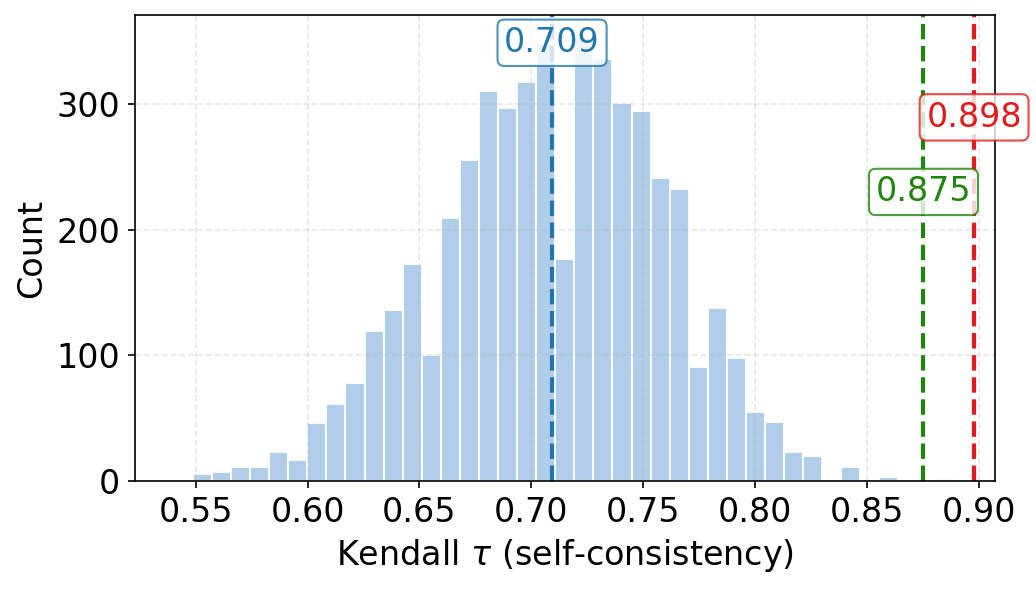} &
        \includegraphics[width=0.31\textwidth]{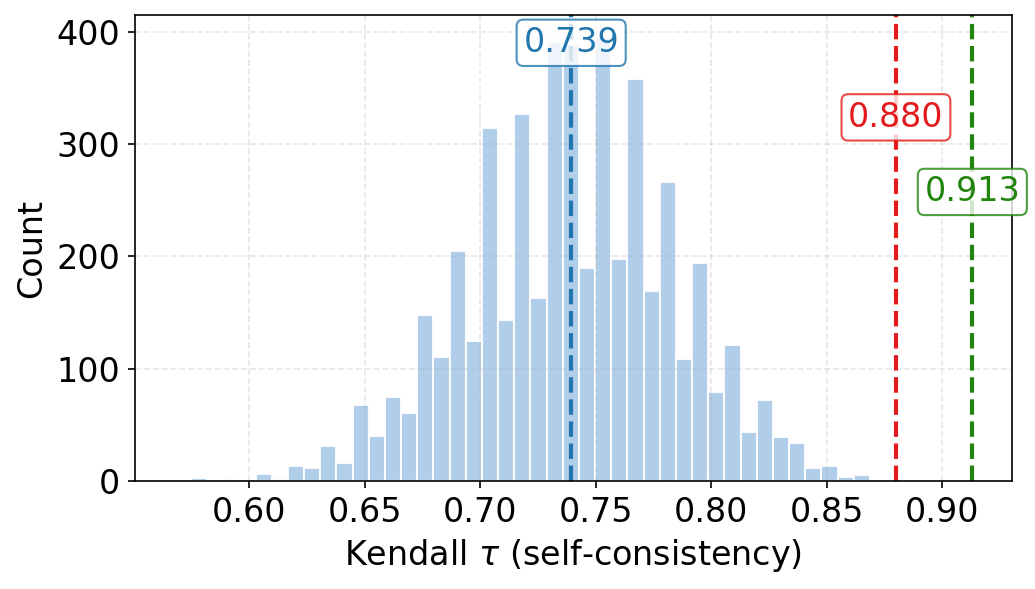} &
        \includegraphics[width=0.31\textwidth]{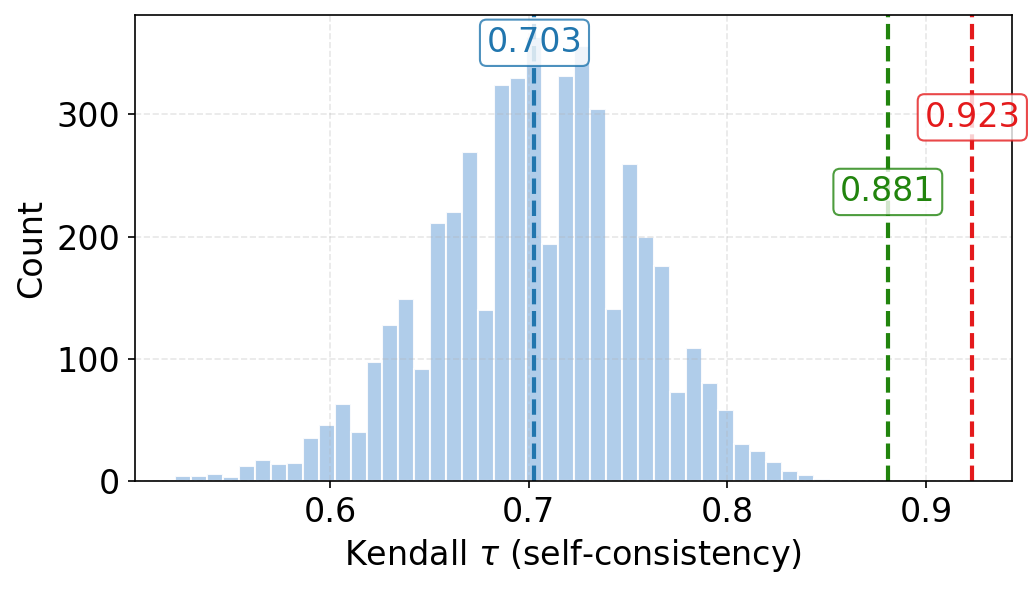} \\[2pt]
        \rotatebox{90}{\hspace{16pt}\textbf{Qwen}} &
        \includegraphics[width=0.31\textwidth]{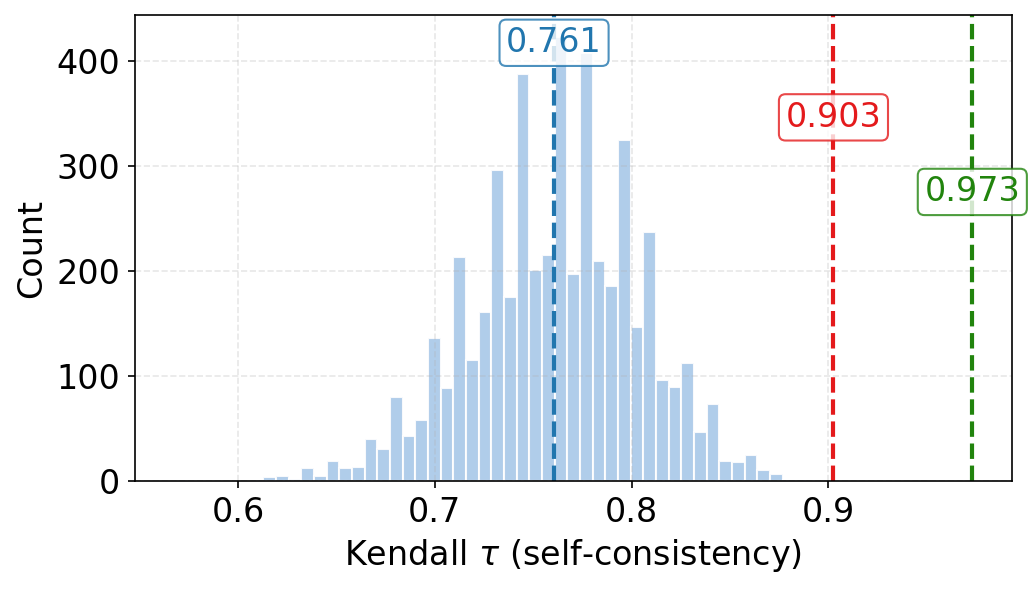} &
        \includegraphics[width=0.31\textwidth]{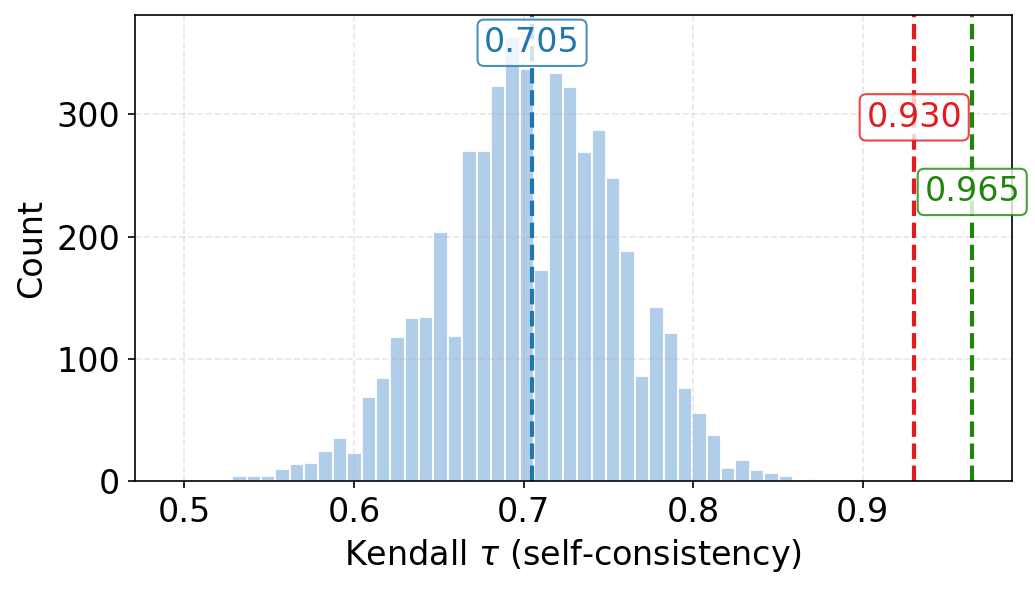} &
        \includegraphics[width=0.31\textwidth]{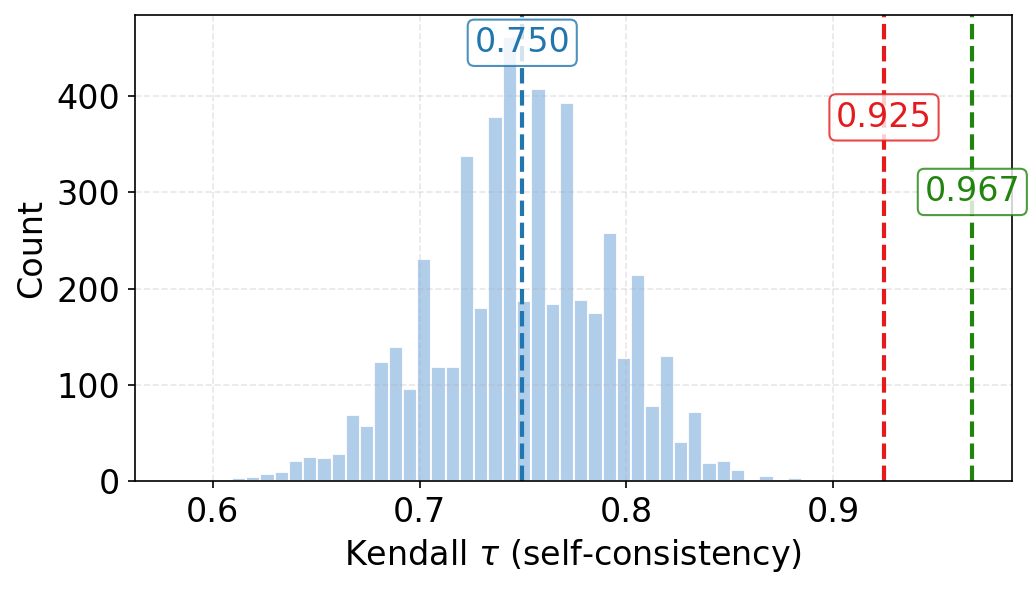} \\
    \end{tabular}
    \caption{Each panel shows the empirical distribution of $\tau_{\mathrm{cons}}$ from 100 bootstrap subsamples of a single $30 \times 30$ complete subtournament (at $p = 0.4$), with the full-pipeline and complete 30x30 tournament $\tau_{\mathrm{cons}}$ overlaid (see legend). Rows correspond to LLM judges (DeepSeek, LLaMA, Qwen); columns correspond to homelessness datasets (VI-SPDAT, VI-F-SPDAT, TAY-VI-SPDAT). In every case, the full-pipeline $\tau_{\mathrm{cons}}$ falls well to the right of the bootstrap distribution, illustrating the distinction between our two reliability axes: the subtournament correctly estimates comparison quality ($\zeta$, see Figure~\ref{fig:ratio-zeta-tau-alpha-sweep}), but inter-run agreement ($\tau_{\mathrm{cons}}$, this Figure) depends on both the quality and the quantity of comparisons collected.}
    \label{fig:bootstrap-grid}
\end{figure*}

\noindent Figure~\ref{fig:expected-r-vs-n} confirms Proposition~\ref{prop:zeta_scaling} visually. For every fixed $\beta$, the inconsistency rate $E[r]$ decreases with $n$ and approaches zero. The intuition is that in the equally-spaced BTL model, the vast majority of triples involve items that are far apart in rank, and those comparisons are nearly deterministic: the upset probability $q_d = \sigma(-\beta d)$ decays exponentially with rank gap $d$. Circular triads then arise almost exclusively from close neighbors, and the number of such close triples grows only linearly in $n$, while the total number of triples grows cubically. So $r$ shrinks as $\Theta(1/n^2)$, driving $\zeta$ toward $1$.

In the language of Section~\ref{subsec:zeta_size_dependence}, this places the BTL model firmly in the decreasing-$r$ regime, and the decrease is fast enough that $\zeta$ not only remains stable but actually increases. However, we note that the scaling of $\zeta$ with $n$ in practice may differ substantially from the stylized model.  If LLM pairwise comparisons behaved like this BTL process with some fixed preference strength, we would expect $\zeta$ to be very high in the tournaments we study ($n$ in the hundreds). As we will soon see, this is not what we observe. The $\zeta$ values in our real-world experiments remain roughly constant across different values of $n$, suggesting that the inconsistency rate $r$ is approximately constant rather than decreasing.  Further, in real prioritization tasks, asking an expensive method like an LLM to perform the 
full $\binom{n}{2}$ comparisons may be infeasible. We now evaluate the relationship between $\zeta$ and the inter-run variability $\tau_{\mathrm{cons}}$.

\subsubsection*{Datasets}
\label{sec:datasets}

We evaluate the diagnostic framework on real-world prioritization tasks spanning two domains: homelessness housing allocation and emergency department triage. 

\noindent \emph{Homelessness housing allocation.} We use household-level vulnerability assessments from St.\ Louis coordinated entry records (2021-2024) \cite{stereetLevelAI}. Three population-specific instruments are treated as separate datasets: VI-SPDAT for single adults ($n = 325$), VI-F-SPDAT for families ($n = 698$), and TAY-VI-SPDAT for transition-age youth ($n = 561$). Each household has a total acuity score used in practice for prioritization and is assigned to a priority band (low, medium, high) based on established scoring rubrics. For the Stage~1 subtournament, we select $m = 30$ households per instrument with proportional representation across priority bands.
 
\noindent \emph{Emergency department triage.} We use ED visits from MIMIC-IV \citep{mimic_iv} ($N = 180{,}057$ total encounters), restricted to encounters with complete triage features and valid acuity scores. Each visit is characterized by vital signs, chief complaint, medication history, and demographics recorded at arrival. The ground-truth label is the Emergency Severity Index (ESI), a five-level acuity score assigned by the triage nurse \citep{wolf2023esi}. For Stage~1, we select $m = 30$ patients with proportional representation across ESI levels, matching the distribution in the full dataset. For Stage~2, we sample $n = 500$ again matching proportional representation across ESI levels from the eligible pool. Full inclusion criteria and feature definitions are in Appendix~\ref{app:datasets}. 

\subsubsection*{LLMs, Prompting and Implementation Details}

We use three LLMs for our experiments: \textit{Meta-Llama-3-8B-Instruct}, \textit{DeepSeek-R1-Distill-Llama-8B}, and \textit{Qwen2.5-7B-Instruct}. Due to the sensitive nature of the data, we deploy these open source models locally with weights downloaded from Hugging-face. The prompts were iteratively designed to try and achieve high consistency, and are provided in Appendix \ref{app:llmprompts}.
We ran these models on one NVIDIA A100 GPU with 80 GB of memory and sampled outputs with 0.1 temperature. 

\begin{figure*}[h!]
    \centering

    \includegraphics[width=0.98\textwidth, trim=3in 6in 4.5in 17in,
        clip]{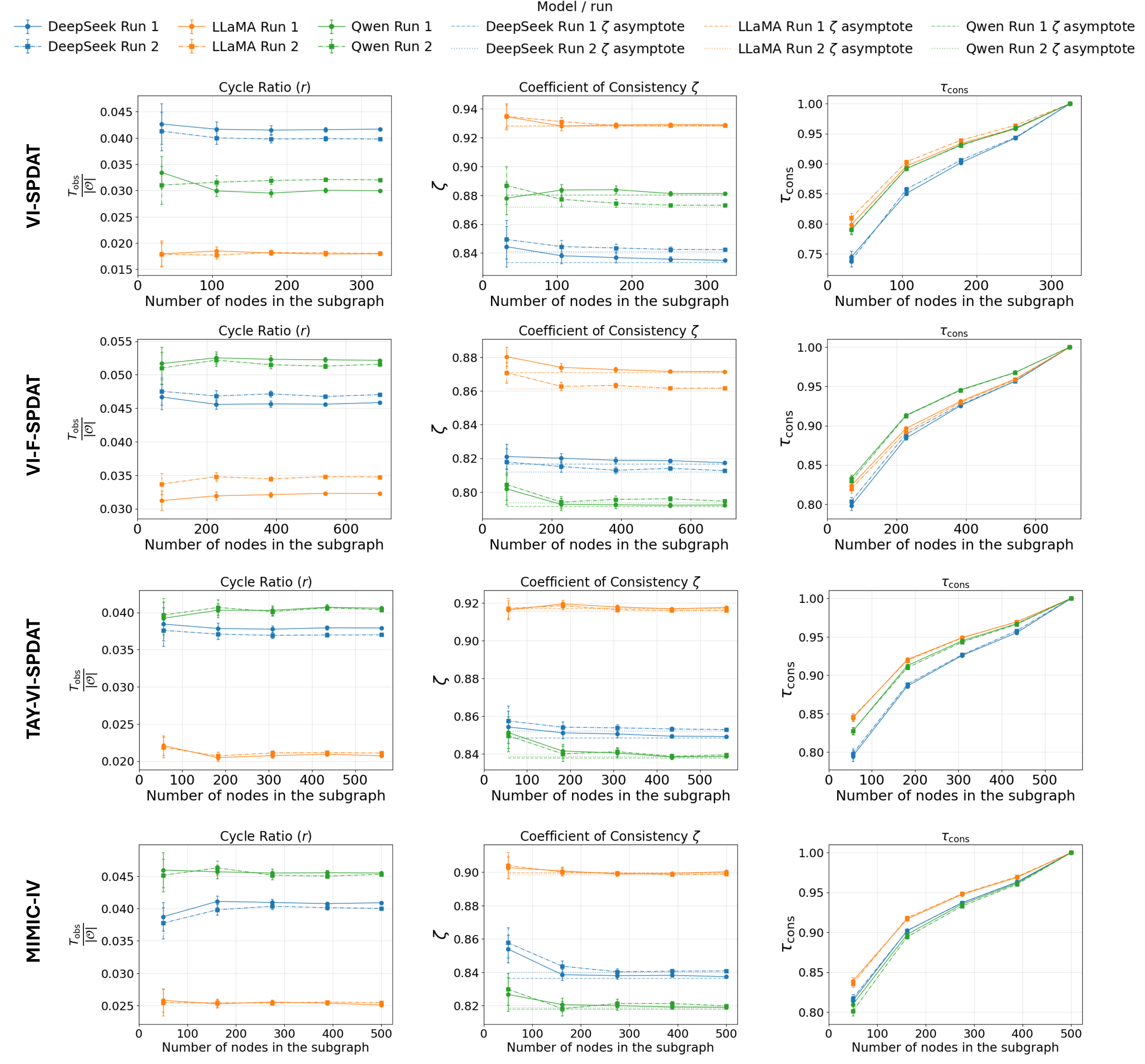}

    \caption{Cycle ratio, $\zeta$, and Rank Centrality (RC) Kendall $\tau$ as functions of the number of nodes sampled in the subgraph. Each model uses a consistent color across runs, while separate runs are distinguished by line style and marker. The two left columns show that $r= \frac{T_{obs}}{|\mathcal{O}|} $ and $\zeta$ computed from the comparison graphs corresponding to a sub-sample of items remain relatively unchanged, making it a good estimator of LLM comparison consistency even with low comparison samples. The right column shows the same cannot be stated for inter-rank consistency ( Kendall's $\tau$ =) between generated ranks using rank centrality from same comparison graphs, making it a relatively poor estimator of consistency without sufficiently large sample size.} 
    \label{fig:ratio-zeta-tau-alpha-sweep}
    \vspace{-1em}
\end{figure*}

\subsection{Experimental Results}
\label{sec:experiments}

\subsubsection{Subtournament Evaluation}
\label{sec:stage1}

For each dataset, we select $m = 30$ representative items and run all $\binom{30}{2} = 435$ pairwise comparisons per LLM judge, forming a complete subtournament. Item order in each comparison is chosen at random. We compute $\zeta$ exactly from the out-degree formula (Equation~\ref{eq:Tn}) and measure self-consistency $\tau_{\mathrm{cons}}$ by running the aggregator on two independent sets of comparisons drawn from the same judge. In addition, we apply the following bootstrap procedure. We repeatedly subsample $40\%$ of the edges of this subtournament. This matches the  observation probability $p = 0.4$ in what we call the \emph{full pipeline}, where all nodes are known, $40\%$ of edges are sampled, and a full ranking constructed. For each run, we record $\zeta$, and in addition the $\binom{100}{2} = 4,950$ pairwise $\tau_{\mathrm{cons}}$ values form an empirical distribution of self-consistency.

Figure~\ref{fig:bootstrap-grid} shows the result of this experiment for inter-run variability. In almost every case, the full-pipeline $\tau_{\mathrm{cons}}$, and the $\tau$ from the Stage~1 fall well to the right of the bootstrap distribution: the full pipeline is considerably more self-consistent than the subtournament. Meanwhile, the subtournament $\hat{\zeta}$ correctly estimates comparison quality (when scaled appropriately by Equation \ref{eq:zeta_decomp}), as in the synthetic sparse-observation experiment (Figure~\ref{fig:sparse-observation}): $\hat{\zeta}$ is unbiased regardless of observation budget.  Figure~\ref{fig:ratio-zeta-tau-alpha-sweep} confirms this in the real data: as the number of items grows, $\zeta$ behaves as predicted by Equation~\eqref{eq:zeta_limit} (the constant $r$ regime) while $\tau_{\mathrm{cons}}$ increases steadily across all datasets and LLM judges. The full pipeline observes more comparisons per item than any single bootstrap subsample, so its rankings are more stable even though the underlying comparison noise is the same.

\begin{table}[ht]
\centering
\small
\setlength{\tabcolsep}{4pt}
\begin{tabular}{lcccccc}
\toprule
& \multicolumn{2}{c}{DeepSeek}
& \multicolumn{2}{c}{LLaMA}
& \multicolumn{2}{c}{Qwen} \\
\cmidrule(lr){2-3}
\cmidrule(lr){4-5}
\cmidrule(lr){6-7}
Dataset
& $\bar{\zeta}$ & $\tau_{\mathrm{cons}}$
& $\bar{\zeta}$ & $\tau_{\mathrm{cons}}$
& $\bar{\zeta}$ & $\tau_{\mathrm{cons}}$ \\
\midrule
TAY-VI-SPDAT & $\circ$ & $-$     & $+$ & $\circ$ & $-$     & $+$ \\
VI-F-SPDAT   & $\circ$ & $\circ$ & $+$ & $-$     & $-$     & $+$ \\
VI-SPDAT     & $-$     & $-$     & $+$ & $\circ$ & $\circ$ & $+$ \\
MIMIC-IV     & $\circ$ & $\circ$ & $+$ & $-$     & $-$     & $+$ \\
\bottomrule
\end{tabular}
\caption{Relative model ranking by average $\zeta$ and Kendall rank correlation ($\tau_{\mathrm{cons}}$). $+$ indicates best, $\circ$ middle, and $-$ worst among the three models for that dataset and measure.}
\label{tab:hmls_mimic_zeta_tau_relative_symbols}
\end{table}

\subsubsection{Full Sparse Pipeline: $\zeta$ and $\tau_{obs}$ as Functions of $n$}
\label{sec:stage2}

For each dataset, we run the full pipeline on all $n$ items with edge sampling probability $p = 0.4$. This is to keep the number of calls to the LLM tractable. We estimate $\hat{\zeta}$ from the fully-observed triplets in the sparse data (Equation~\eqref{eq:zeta_hat}) and measure self-consistency by running the pipeline twice on independent edge samples. To estimate how our two measures scale with $n$, for each dataset, we take the comparisons performed above, randomly sample fractions of nodes $\in \{0.1, 0.325, 0.55, 0.775\}$, and construct a comparison graph with only the edges that had been sampled. We repeat this 100 times and empirically compute the expected $\tau, \zeta$ and $r= \frac{T_{obs}}{|\mathcal{O}|}$. The results are provided in Figure \ref{fig:ratio-zeta-tau-alpha-sweep}.

As $n$ increases, the inter-run variability of the ranks produced also decreases. This is because $\tau_{\text{cons}}$ depends on the number of effective comparisons per item, which increases as we get more nodes with the same sampling ratio of edges. However, the expected inconsistency rate $r$ remains approximately constant across values of $n$, placing these datasets in the constant-$r$ regime of the decomposition in Equation~\eqref{eq:zeta_decomp}. This is qualitatively different from the BTL model, where Proposition~\ref{prop:zeta_scaling} predicts $r = \Theta(1/n^2)$ and $\zeta \to 1$.

\subsection{Consequences}

The discussion above demonstrates in practice that $\zeta$ and $\tau_{\mathrm{cons}}$ are different ways of measuring consistency of a ranking method, and they can behave quite differently. Another natural question is whether a method that is good along one dimension is naturally also good along another. Our experiments also allow us to examine this question in the context of the three LLMs that we use. Table \ref{tab:hmls_mimic_zeta_tau_relative_symbols} shows which of the models performs best/middle/worst along each of these dimensions across the four datasets we look at.\footnote{
For complete results see Appendix~\ref{appen:sanmay_asked} 
}
We find that no model Pareto-dominates across both dimensions. In fact, Qwen is reliably best in terms of inter-run variability, but also reliably worse in terms of intra-run consistency! Meanwhile, LLaMA does best across all the datasets in terms of intra-run consistency, but it is twice worst and twice in the middle in terms of inter-run variability. DeepSeek typically occupies the middle position along both axes. Which dimension is more important? This will of course come down to a question of what a stakeholder values in a particular domain. But our results illustrate clearly the value of evaluating ranking along both these dimensions of consistency.

\section{Discussion}
\label{sec:discussion}

The central finding of this paper is that intra-run consistency and inter-run variance are genuinely independent axes of ranking reliability, not just in theory (Table~\ref{tab:ranking-methods}) but in practice. The synthetic calibration shows that even under the idealized BTL model, where a single parameter governs all comparison noise, the two axes are correlated but not redundant: tournaments with identical $\zeta$ can have very different $\tau_{\text{cons}}$ depending on the number of comparisons per item. The real-world experiments sharpen this separation further. The inconsistency rate $r$ remains approximately constant as $n$ grows, placing all four datasets in the constant-$r$ regime of our decomposition and confirming that LLM comparison noise does not have the rank-gap-dependent structure of the BTL model (where Proposition~\ref{prop:zeta_scaling} predicts $r = \Theta(1/n^2)$). Most strikingly, no LLM dominates across both dimensions: Qwen produces the most stable rankings across runs but the least internally consistent comparisons, while LLaMA exhibits the opposite profile. This means that choosing a model on the basis of only one diagnostic can be actively misleading. A practitioner who evaluated only inter-run variance would select Qwen; one who evaluated only intra-run consistency would select LLaMA. Neither choice is wrong, but neither is complete.
 
\noindent Our results suggest a simple diagnostic protocol. First, run a small complete subtournament and compute $\zeta$ exactly; because $\zeta$ is unbiased regardless of observation budget, this inexpensive step reliably estimates comparison quality before committing to the full pipeline. If $\zeta$ is low, the comparisons themselves are inconsistent in the sense of transitivity, and no amount of additional data will ``fix'' that issue. If $\zeta$ is high but $\tau_{\text{cons}}$ is low, the comparisons are consistent but too few have been collected, and the remedy is to increase the sampling fraction rather than change the model. This diagnostic separation between signal quality and signal quantity is the practical payoff of the $2 \times 2$ taxonomy.

\section*{Acknowledgements}
We are grateful for support from NSF Award 2533162. We also thank the various community partners who helped conceptualize the challenges facing the delivery of homeless services, as well as their ongoing efforts to support local families.

\bibliographystyle{unsrtnat}
\bibliography{references}

\newpage
\appendix
\onecolumn
\section{Additional Details on the Synthetic Model/Experiments}
\label{appen:synthetic}

This appendix provides implementation details and supplementary experiments for the synthetic experiments in Section~\ref{sec:synthetic}. Throughout, error bars on the $y$-axis represent $\pm 1$ standard deviation across 100 independent trials. Error bars on the $x$-axis, where present, represent the standard deviation of the observed $\zeta$ around its analytic expectation $\mathbb{E}[\zeta]$.

\begin{remark}[Regularization of RC at High $\beta$]
    \label{rem:damping}
    RC ranks items by the stationary distribution of a Markov chain whose transition matrix $\bm{P}$ is constructed from the tournament outcomes. In the single-comparison setting, an item that wins all $n - 1$ of its comparisons becomes an absorbing state: $P_{ii} = 1$ and the stationary distribution places all mass there, assigning zero score to every other item. Under the equally-spaced BTL model this occurs with probability $\prod_{j=1}^{n-1} \sigma(\beta j)$, which approaches 1 for moderate $\beta$. Following the standard PageRank remedy \citep{brin1998anatomy}, we replace $\bm{P}$ with $\bm{P}' = (1 - \alpha)\bm{P} + (\alpha/n)\bm{J}$, where $\bm{J}$ is the all-ones matrix and $\alpha = 0.01$. This eliminates absorbing states while leaving rankings in the moderate-noise regime essentially unchanged.
\end{remark}

\subsection{Matching $\zeta$ across $n$}
\label{appen:matching}
 
To ensure that calibration curves for different values of $n$ are evaluated at the same set of $\zeta$ values, we solve for the value of $\beta$ that yields a given target $\mathbb{E}[\zeta]$. Under the equally-spaced BTL model (Definition~\ref{def:esbtl}), the expected circular triad count admits an exact closed form:
\begin{equation}
\label{eq:expected-Tn}
\mathbb{E}[T_n] = \sum_{i < j < k} \bigl[ P_{ij}\, P_{jk}\, P_{ki} \;+\; P_{ji}\, P_{kj}\, P_{ik} \bigr],
\end{equation}
where $P_{ij} = \sigma\bigl(\beta(j - i)\bigr)$ is the pairwise comparison probability from Equation~\eqref{eq:pij}. The expected coefficient of consistency is then $\mathbb{E}[\zeta] = 1 - \mathbb{E}[T_n] / T_{\max}$. Because the mapping $\beta \mapsto \mathbb{E}[\zeta]$ is smooth and strictly increasing on $[0, \infty)$, ranging from $0$ as $\beta \to 0$ to $1$ as $\beta \to \infty$, there exists a unique $\beta^*$ satisfying $\mathbb{E}[\zeta(\beta^*, n)] = \zeta^*$ for every target $\zeta^* \in (0,1)$ and every $n$. We find $\beta^*$ via Brent's root-finding method~\citep{brent1973algorithms} applied to $\mathbb{E}[\zeta(\beta, n)] - \zeta^* = 0$, with tolerance $10^{-8}$. This allows all experiments to be evaluated at a common grid of $\zeta$ values $(0.1, 0.2, \ldots, 0.9)$, eliminating the confound of different $n$ values producing different $\zeta$ ranges at the same $\beta$.

\subsection{Proof of Proposition~\ref{prop:zeta_scaling}}
\label{appen:zeta_scaling_proof}
\zetascaling*

\begin{proof}
\textit{Triad anatomy.} For a triple $i < j < k$ with rank gaps $a = j - i$ and $b = k - j$, the third pairwise distance is $a + b$. Let $q_d = \sigma(-\beta d)$ be the upset probability at gap $d$. Under BTL the three pair outcomes are independent. The triad is cyclic in exactly two configurations:
\begin{itemize}
\item single long-edge upset: $i$ beats $j$, $j$ beats $k$, $k$ beats $i$;
\item two short-edge upsets: $j$ beats $i$, $k$ beats $j$, $i$ beats $k$.
\end{itemize}
Every other configuration is transitive. Therefore: 
    \begin{equation*}
        P(\mathrm{cycle} \mid a, b) = (1 - q_a)(1 - q_b) q_{a+b} + q_a q_b (1 - q_{a+b}).
    \end{equation*}

\noindent \textit{Upper bound on $E[T_n]$.} Since $\sigma(-x) \le e^{-x}$ for $x \ge 0$,
    \begin{equation*}
        P(\mathrm{cycle} \mid a, b) \le q_{a+b} + q_a q_b \le 2 e^{-\beta(a+b)}.
    \end{equation*}
The number of triples with gap pattern $(a, b)$ is $n - a - b$, so: 
    \begin{equation*}
        E[T_n] = \sum_{\substack{a, b \ge 1 \\ a + b \le n - 1}} (n - a - b) P(\mathrm{cycle} \mid a, b) \le 2 n \sum_{a, b \ge 1} e^{-\beta(a + b)} = 2 n \left(\frac{e^{-\beta}}{1 - e^{-\beta}}\right)^{ 2} = O(n),
    \end{equation*}
with constant depending only on $\beta$. \\

\noindent \textit{Lower bound on $E[T_n]$.} Consecutive triples have $a = b = 1$ and there are $n - 2$ of them. For each,
\begin{equation*}
    P(\mathrm{cycle} \mid 1, 1) \ge q_1^2 (1 - q_2),
\end{equation*}
a strictly positive constant in $\beta$. Hence $E[T_n] \ge (n - 2) q_1^2 (1 - q_2) = \Omega(n)$. \\

\noindent \textit{Denominator.} $T_{\max} = (n^3 - n)/24$ for odd $n$ and $(n^3 - 4n)/24$ for even $n$, both $\Theta(n^3)$. Combining, $E[1 - \zeta] = \Theta(n) / \Theta(n^3) = \Theta(1/n^2)$.
\end{proof}

\paragraph{Why the upward shift in Figure 1b with $n$} The expected number of circular triads is carried by triples with small rank gaps, since the cycle probability is exponentially small in the long-edge gap $a + b$. Small-gap triples are linear in $n$, while the total triad count is cubic. $\zeta$ therefore mixes a linear numerator with a cubic denominator, and at fixed $\beta$ it approaches $1$ at rate $1/n^2$. Matching two tournaments at the same value of $\zeta$ across different $n$ is consequently matching them on a statistic whose denominator absorbs $\Theta(n^3)$ trivially-transitive triples; the larger-$n$ tournament must have substantially smaller $\beta$ (noisier per-pair comparisons) to reach the same $\zeta$. Self-consistency, by contrast, is controlled by the $n - 1$ comparisons each item participates in and tracks per-pair noise directly through aggregation accuracy.

\subsection{Aggregator Comparison}
\label{appen:aggregator_comparison}

Here, we test whether the diagnostic value of $\zeta$ is specific to Rank Centrality or generalizes across aggregation methods. We fix $n = 50$ and run three aggregators on the same sampled tournaments: Rank Centrality, Bradley--Terry Maximum Likelihood Estimation (BT-MLE, via the iterative Luce spectral ranking algorithm), and Borda count (ranking by total win count). For each aggregator, we record accuracy $\tau_{\mathrm{acc}}$ and self-consistency $\tau_{\mathrm{cons}}$ at each target $\zeta$ value.

\begin{figure*}[h!]
    \centering
    \subfloat[Accuracy]{%
        \includegraphics[width=0.48\textwidth]{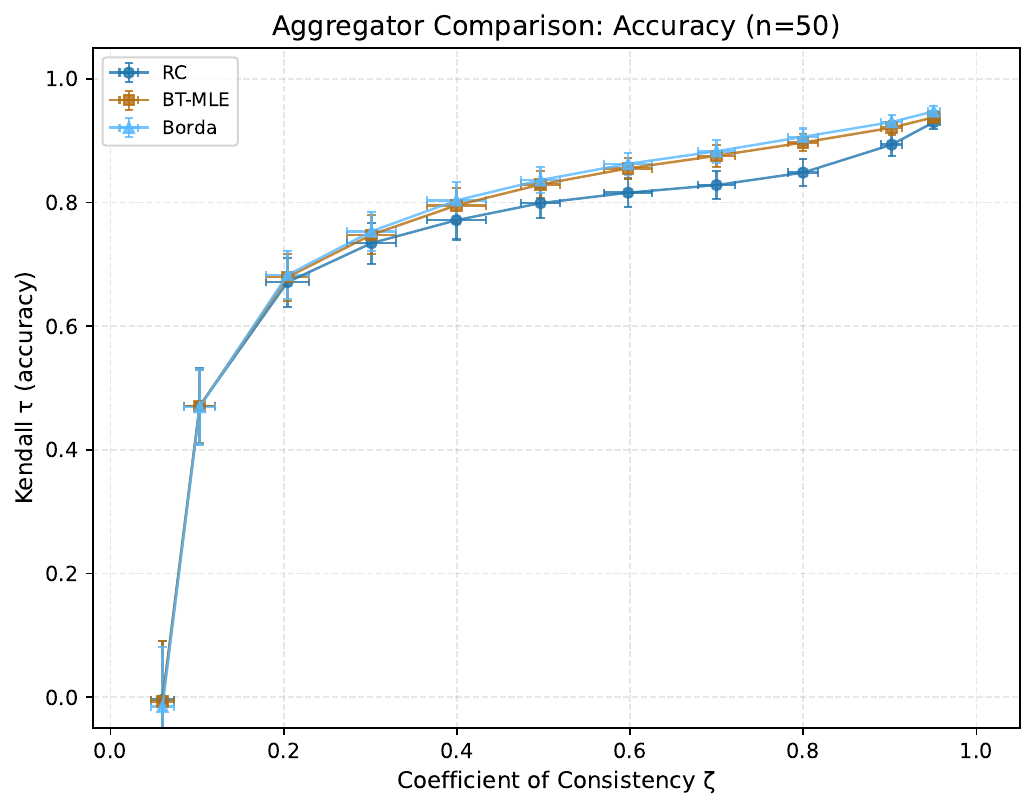}%
        \label{fig:exp2-accuracy}}
    \hfill
    \subfloat[Self-consistency]{%
        \includegraphics[width=0.48\textwidth]{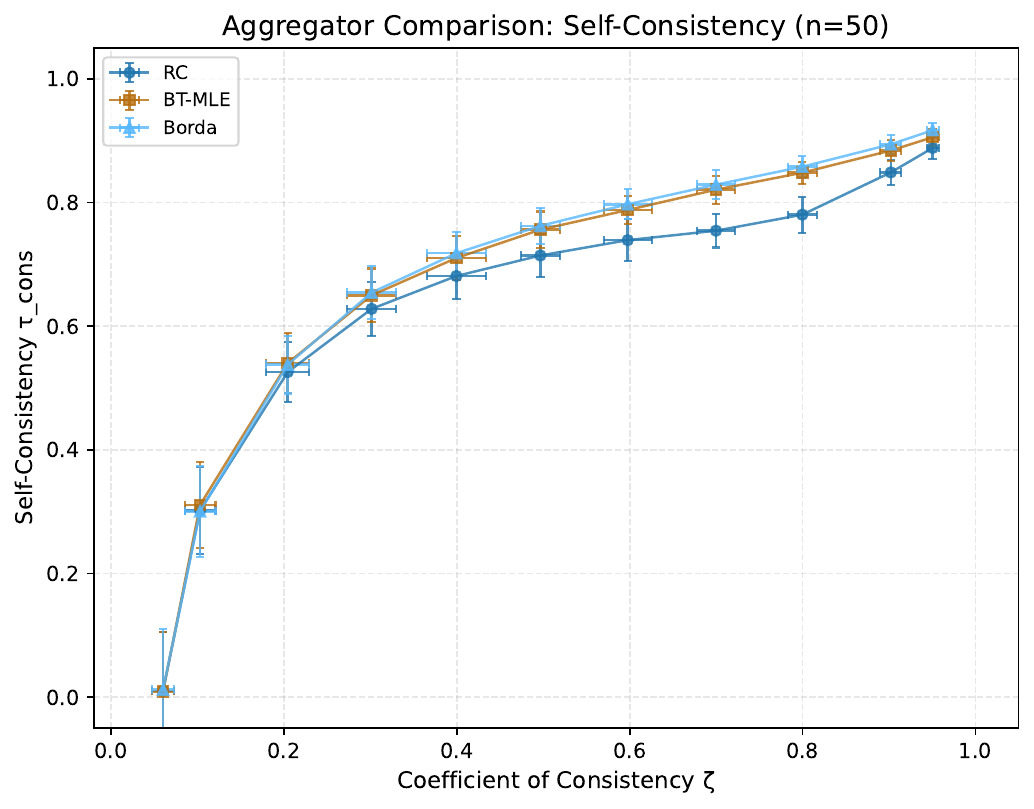}%
        \label{fig:exp2-consistency}}
    \caption{Aggregator comparison ($n = 50$). (a)~$\zeta$ vs.\ accuracy for RC, BT-MLE, and Borda; all three aggregators yield nearly identical accuracy as a function of $\zeta$, indicating that comparison quality dominates the choice of aggregation method. (b)~$\zeta$ vs.\ self-consistency for the same three aggregators; the curves are again nearly coincident.}
    \label{fig:aggregator_comparison}
\end{figure*}

\noindent Figure~\ref{fig:exp2-accuracy} shows that all three aggregators produce nearly identical accuracy curves as a function of $\zeta$. Figure~\ref{fig:exp2-consistency} shows the same for self-consistency. This establishes that $\zeta$ is a diagnostic for the \emph{data}, not for any particular aggregation algorithm -- the quality of the pairwise comparisons, as captured by $\zeta$, is the dominant factor, and the choice of aggregator is secondary.

\section{Dataset Details}
\label{app:datasets}

Both experimental domains share a common structure: each case is characterized by a set of features, a human expert assigns a discrete priority level, and the task is to decide which of two cases should be prioritized for intervention. The domains differ in population, feature representation, and label granularity.

\subsection{Homelessness Service Allocation}
\label{data:homelessness}

We use household-level vulnerability assessments from St.\ Louis coordinated entry records (2021-2024) \cite{stereetLevelAI}. The data comprise assessments from three population-specific VI-SPDAT instruments: VI-SPDAT for single adults ($n = 325$), VI-F-SPDAT for families ($n = 698$), and TAY-VI-SPDAT for transition-age youth ($n = 561$). Each record contains the full set of raw questionnaire responses (35 items for single adults, 54 for families, 41 for youth; see Appendix~\ref{appendix:ques} for the full set of questions) together with the corresponding total acuity score used in practice for prioritization. We treat the three instruments as separate datasets throughout, since they correspond to different populations with different questionnaires and scoring rubrics.

Each household is assigned to a priority band based on its total score. For single adults and youth, the bands are low ($0$--$3$), medium ($4$--$7$), and high ($8+$). For families, the bands are low ($0$--$3$), medium ($4$--$8$), and high ($9+$), following the established scoring rubrics~\citep{orgcode2015vispdat, orgcode2015vifspdat, orgcode2015tayvispdat}. Table~\ref{tab:band-thresholds-hmls} summarizes the thresholds, pool sizes, and the subtournament samples drawn with proportional representation across bands.

\begin{table}[h!]
\centering
\begin{tabular}{llcccc}
\toprule
Assessment & Band & Threshold & Pool & Pool \% & Sample \\
\midrule
\multirow{3}{*}{VI-SPDAT ($n = 325$)}
  & Low    & $0$--$3$ &  61 & 18.8\% &  6 \\
  & Medium & $4$--$7$ & 175 & 53.8\% & 16 \\
  & High   & $8+$     &  89 & 27.4\% &  8 \\
\addlinespace
\multirow{3}{*}{VI-F-SPDAT ($n = 698$)}
  & Low    & $0$--$3$ &  45 &  6.4\% &  2 \\
  & Medium & $4$--$8$ & 415 & 59.5\% & 18 \\
  & High   & $9+$     & 238 & 34.1\% & 10 \\
\addlinespace
\multirow{3}{*}{TAY-VISPDAT ($n = 561$)}
  & Low    & $0$--$3$ &  49 &  8.7\% &  3 \\
  & Medium & $4$--$7$ & 281 & 50.1\% & 15 \\
  & High   & $8+$     & 231 & 41.2\% & 12 \\
\bottomrule
\end{tabular}
\caption{Priority band score thresholds, pool sizes, and subtournament sample sizes ($m = 30$) by assessment type. Samples are drawn with proportional representation across priority bands.}
\label{tab:band-thresholds-hmls}
\end{table}

\subsection{Emergency Department Triage}
\label{data:mimic}

We use emergency department triage data from MIMIC-IV~\citep{mimic_iv}, a publicly available de-identified electronic health record database from Beth Israel Deaconess Medical Center in Boston. We link three MIMIC-IV modules: the ED module~\citep{mimic_iv_ed} (triage assessments and medication reconciliation), the hospital module (patient demographics), and the admissions table (supplementary sociodemographic fields for admitted patients).

\begin{table*}[h!]
\centering
\begin{tabular}{clp{5cm}rrrr}
\toprule
\textbf{ESI} & \textbf{Category} & \textbf{Definition} & \textbf{Pool} & \textbf{Pool \%} & \textbf{$m = 30$} & \textbf{$n = 500$} \\
\midrule
1 & Immediate   & Requires immediate life-saving intervention (e.g., cardiac arrest, respiratory failure)  & 5,600  &  3.1\% &  1 &  16 \\
2 & Emergent    & High-risk situation; patient should not wait (e.g., severe pain, altered mental status)   & 60,280 & 33.5\% & 10 & 167 \\
3 & Urgent      & Stable, but expected to require multiple resources (e.g., labs, imaging, IV medications)  & 98,919 & 54.9\% & 17 & 275 \\
4 & Less urgent & Stable, expected to require one resource (e.g., a single lab test or X-ray)               & 14,726 &  8.2\% &  2 &  41 \\
5 & Non-urgent  & Stable, expected to require no resources (e.g., prescription refill, simple wound check)  &    532 &  0.3\% &  0 &   1 \\
\bottomrule
\end{tabular}
\caption{Emergency Severity Index (ESI) acuity levels, clinical definitions, and sampling details. From $180{,}057$ eligible MIMIC-IV encounters, we draw two cohorts with proportional representation across ESI levels: a subtournament ($m = 30$) and a full pipeline ($n = 500$).}
\label{tab:esi-levels}
\end{table*}

Each record corresponds to a single ED visit. At arrival, triage nurses record temperature, heart rate, respiratory rate, oxygen saturation, systolic and diastolic blood pressure, a patient-reported pain score (0--10), and a free-text chief complaint. Medication reconciliation data document pre-visit medications, serving as a proxy for medical history. Demographics available at triage include age, gender, race, and mode of arrival (e.g., ambulance, walk-in); primary language and marital status are additionally available for admitted patients.

\begin{table}[H]
\centering
\small
\begin{tabular}{llp{4.5cm}}
\toprule
\textbf{Category} & \textbf{Prompt label} & \textbf{Description} \\
\midrule
\multirow{3}{*}{Demographics}
  & \texttt{age}          & Patient age at time of visit \\
  & \texttt{gender}       & Patient gender \\
  & \texttt{race}         & Patient self-reported race \\
\addlinespace
Arrival
  & \texttt{arrival\_mode} & Mode of arrival (e.g., ambulance, walk-in) \\
\addlinespace
\multirow{7}{*}{Vital signs}
  & \texttt{temperature\_in\_fahrenheit}       & Body temperature ($^\circ$F) \\
  & \texttt{heart\_rate\_bpm}                  & Heart rate (beats per minute) \\
  & \texttt{respiratory\_rate}                 & Respiratory rate (breaths per minute) \\
  & \texttt{o2\_saturation\_percent}           & Peripheral oxygen saturation (\%) \\
  & \texttt{systolic\_blood\_pressure\_mmhg}   & Systolic blood pressure (mmHg) \\
  & \texttt{diastolic\_blood\_pressure\_mmhg}  & Diastolic blood pressure (mmHg) \\
  & \texttt{pain\_score}                       & Patient-reported pain (0--10) \\
\addlinespace
Clinical presentation
  & \texttt{chief\_complaint}          & Free-text reason for presenting \\
\addlinespace
\multirow{3}{*}{Medications}
  & \texttt{current\_medications}      & Medications taken prior to visit \\
  & \texttt{medication\_classes}       & Therapeutic classes of medications \\
  & \texttt{num\_current\_medications} & Number of current medications \\
\bottomrule
\end{tabular}
\caption{Features presented to the LLM for each patient in pairwise triage comparisons. Field names are mapped from the original MIMIC-IV schema to the labels used in the prompt.}
\label{tab:mimic-features}
\end{table}

\noindent The ground-truth label is the Emergency Severity Index (ESI) acuity score, an integer from 1 (most acute) to 5 (least acute) assigned by the triage nurse. ESI is a validated five-level triage algorithm widely used in U.S.\ emergency departments~\citep{wolf2023esi}. Levels 1 and 2 are determined by patient acuity (immediate life-saving need vs.\ high-risk, should not wait); Levels 3 through 5 are distinguished by the number of resources the patient is expected to require. Table~\ref{tab:esi-levels} summarizes the ESI levels, pool sizes, and the two cohorts sampled with proportional representation. We restrict the sample to ED visits with a valid acuity score, excluding encounters in which the patient eloped, left without being seen, or left against medical advice. We require all triage feature fields to be non-missing and retain only the most recent visit per patient to avoid within-patient correlation. The resulting eligible pool contains $180{,}057$ unique patients.

\subsection{Experimental Data underlying Table \ref{tab:hmls_mimic_zeta_tau_relative_symbols}}
\label{appen:sanmay_asked}

\begin{table}[H]
\centering
\scriptsize
\resizebox{0.5\linewidth}{!}{%
\begin{tabular}{lrrrr}
\toprule
LLM & $\zeta_1$ & $\zeta_2$ & $\bar{\zeta}$ & $\tau_{\mathrm{cons}}$ \\
\midrule
DeepSeek & 0.849144 & 0.852818 & 0.850981 & 0.887179 \\
LLaMA    & 0.917566 & 0.916106 & 0.916836 & 0.923300 \\
Qwen     & 0.838598 & 0.839361 & 0.838979 & 0.924599 \\
\bottomrule
\end{tabular}%
}
\caption{$\zeta_1$, $\zeta_2$, mean $\bar{\zeta}$, and Kendall $\tau_{\mathrm{cons}}$ for TAY-VI-SPDAT.}
\label{tab:zeta_tau_tayvispdat}
\end{table}

\begin{table}[H]
\centering
\scriptsize
\resizebox{0.5\linewidth}{!}{%
\begin{tabular}{lrrrr}
\toprule
LLM & $\zeta_1$ & $\zeta_2$ & $\bar{\zeta}$ & $\tau_{\mathrm{cons}}$ \\
\midrule
DeepSeek & 0.812661 & 0.817402 & 0.815031 & 0.884729 \\
LLaMA    & 0.871484 & 0.861638 & 0.866561 & 0.879903 \\
Qwen     & 0.792314 & 0.794580 & 0.793447 & 0.930426 \\
\bottomrule
\end{tabular}%
}
\caption{$\zeta_1$, $\zeta_2$, mean $\bar{\zeta}$, and Kendall $\tau_{\mathrm{cons}}$ for VI-F-SPDAT.}
\label{tab:zeta_tau_vifspdat}
\end{table}

\begin{table}[H]
\centering
\scriptsize
\resizebox{0.5\linewidth}{!}{%
\begin{tabular}{lrrrr}
\toprule
LLM & $\zeta_1$ & $\zeta_2$ & $\bar{\zeta}$ & $\tau_{\mathrm{cons}}$ \\
\midrule
DeepSeek & 0.834824 & 0.842320 & 0.838572 & 0.852877 \\
LLaMA    & 0.928576 & 0.928845 & 0.928711 & 0.897892 \\
Qwen     & 0.881298 & 0.873071 & 0.877184 & 0.902564 \\
\bottomrule
\end{tabular}%
}
\caption{$\zeta_1$, $\zeta_2$, mean $\bar{\zeta}$, and Kendall $\tau_{\mathrm{cons}}$ for VI-SPDAT.}
\label{tab:zeta_tau_vispdat}
\end{table}

\begin{table}[H]
\centering
\scriptsize
\resizebox{0.5\linewidth}{!}{%
\begin{tabular}{lrrrr}
\toprule
LLM & $\zeta_1$ & $\zeta_2$ & $\bar{\zeta}$ & $\tau_{\mathrm{cons}}$ \\
\midrule
DeepSeek & 0.837359 & 0.840854 & 0.839107 & 0.896449 \\
LLaMA    & 0.900303 & 0.898932 & 0.899618 & 0.894846 \\
Qwen     & 0.819045 & 0.819877 & 0.819461 & 0.897267 \\
\bottomrule
\end{tabular}%
}
\caption{$\zeta_1$, $\zeta_2$, mean $\bar{\zeta}$, and Kendall $\tau_{\mathrm{cons}}$ for MIMIC.}
\label{tab:zeta_tau_mimic}
\end{table}

\section{LLM Prompts}
\label{app:llmprompts}

Each pairwise comparison is presented to the LLM as a single-turn conversation. The prompt contains the full feature representation of both cases (questionnaire responses for homelessness, triage features for ED visits) and asks the model to select which case should be prioritized, along with a short justification. We found that eliciting a brief explanation alongside the structured decision improved consistency compared to requesting only the formatted template. The placeholder \texttt{<insert block data>} is replaced with the relevant case data for each comparison.

\subsection{Pairwise Comparison Prompts}

Figures~\ref{fig:prompt-vispdat} and~\ref{fig:prompt-mimic} show the prompts used for the homelessness and ED triage domains respectively. The homelessness prompt is shared across all three VI-SPDAT instruments, with the questionnaire content varying by assessment type.

\begin{figure}[H]
    \centering
    \includegraphics[width=\textwidth]{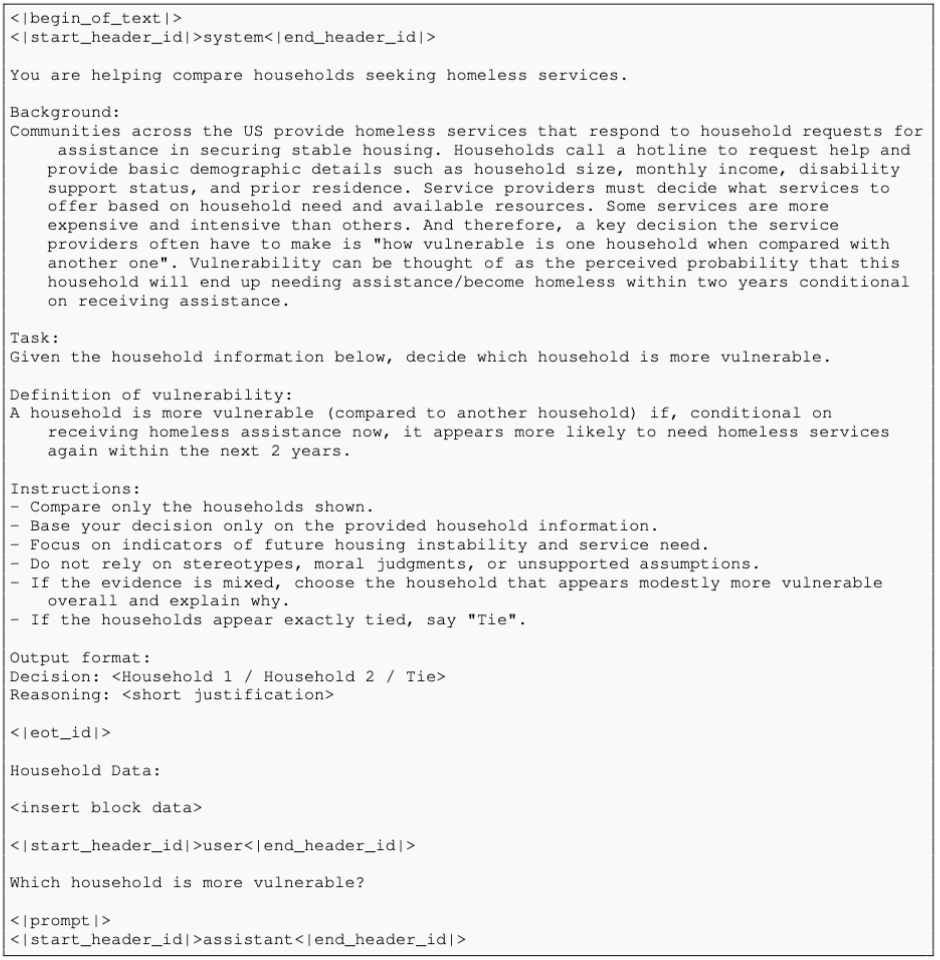}
    \caption{Pairwise comparison prompt for homelessness prioritization (VI-SPDAT, VI-F-SPDAT, and TAY-VI-SPDAT datasets).}
    \label{fig:prompt-vispdat}
\end{figure}

\begin{figure}[H]
    \centering
    \includegraphics[width=\textwidth]{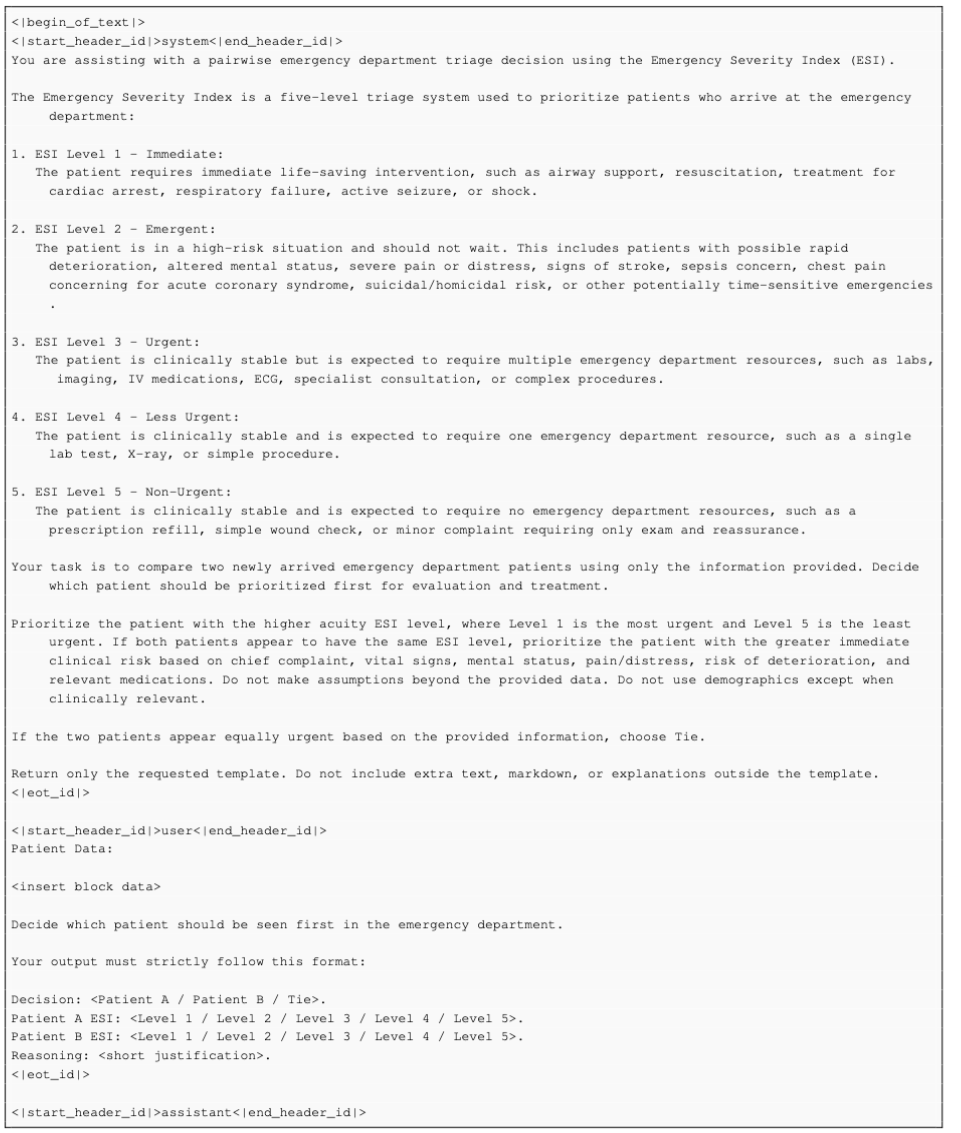}
    \caption{Pairwise comparison prompt for emergency department triage (MIMIC-IV dataset).}
    \label{fig:prompt-mimic}
\end{figure}

\subsection{Second-Level Response Parsing}

In practice, we found that LLMs occasionally produced contradictory outputs: the structured template would indicate one decision while the free-text explanation would argue for the opposite. To resolve these conflicts, we use a second-level LLM (\textit{Qwen2.5-7B-Instruct}) to parse the full output of the decision-making LLM and extract the final prioritization decision. Figures~\ref{fig:prompt-parser-vispdat} and~\ref{fig:prompt-parser-mimic} show the parsing prompts for the homelessness and ED triage domains respectively.
 
\begin{figure}[h!]
    \centering
    \fbox{\includegraphics[width=0.8\textwidth]{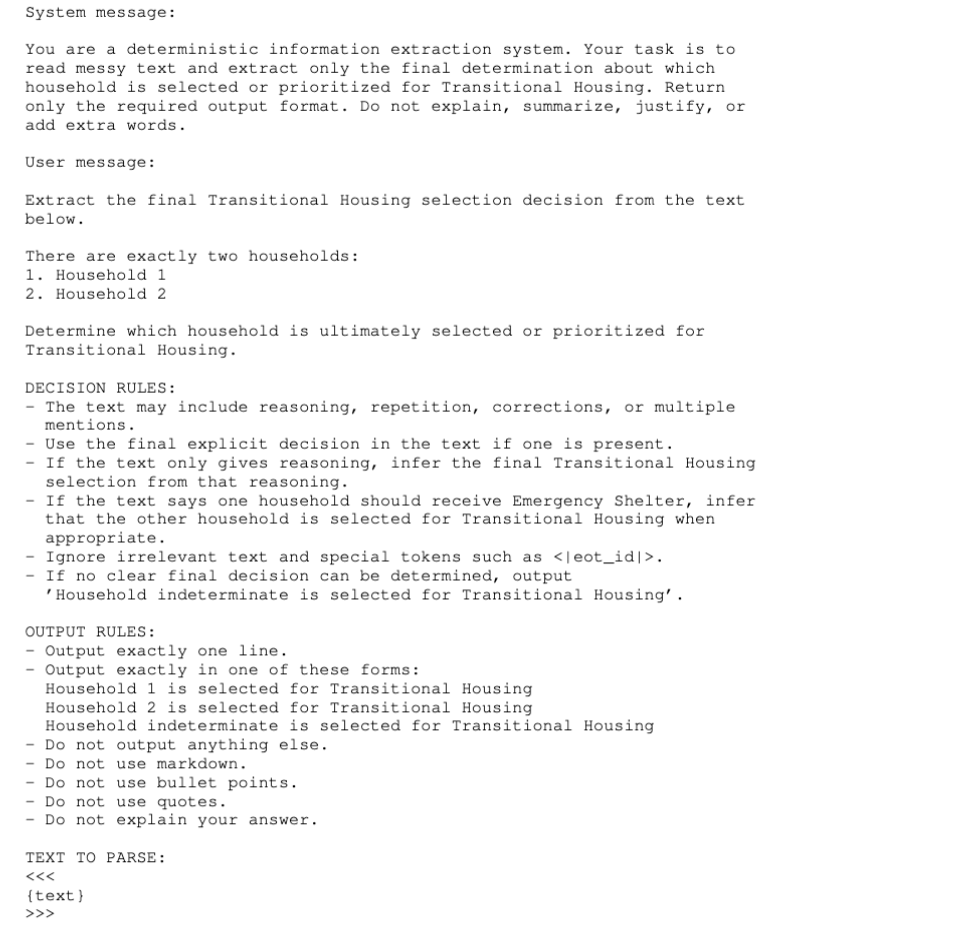}}
    \caption{Second-level parsing prompt for homelessness prioritization outputs.}
    \label{fig:prompt-parser-vispdat}
\end{figure}
 
\begin{figure}[H]
    \centering
    \includegraphics[width=0.8\textwidth]{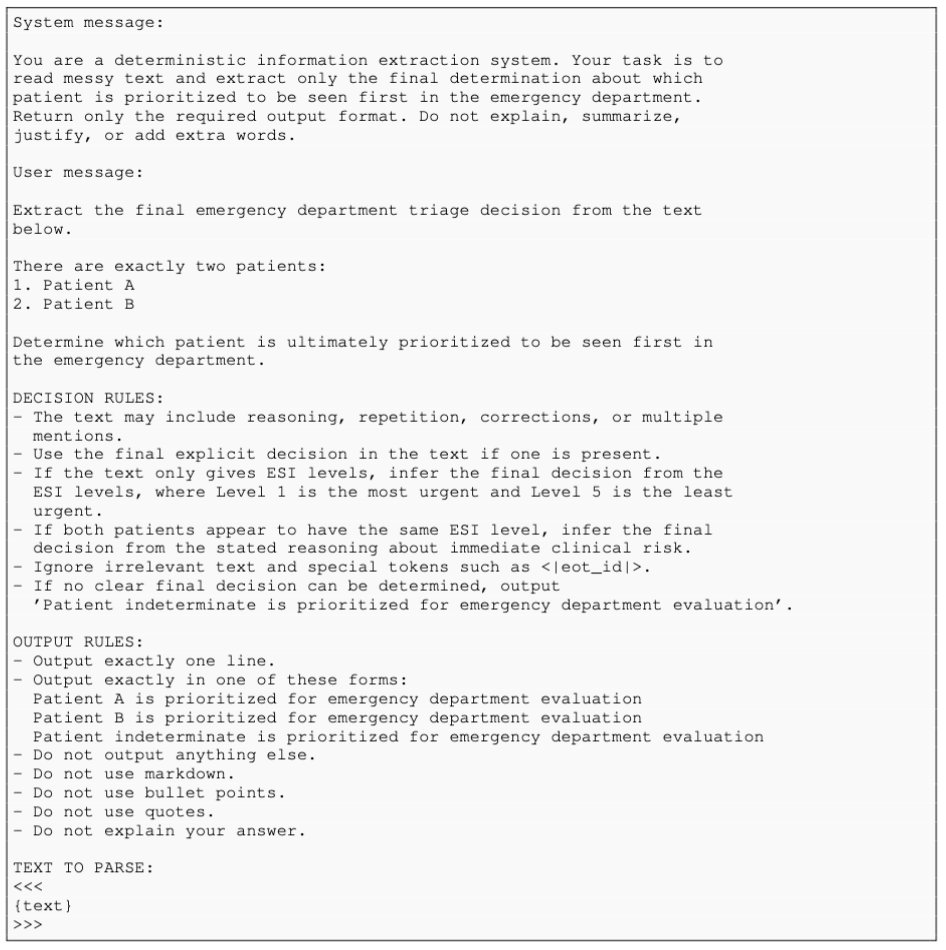}
    \caption{Second-level parsing prompt for emergency department triage outputs.}
    \label{fig:prompt-parser-mimic}
\end{figure}

\section{Questionnaires used in VI-SPDAT system to score vulnerability}
\label{appendix:ques}
\subsection{VI-SPDAT}
\label{appen:visp}
\begin{enumerate}

    \item Are you currently able to care for your basic needs, such as bathing, changing clothes, using the restroom, obtaining food, and accessing clean water?
    \item Are you not taking any medications that a doctor has prescribed for you?
    \item Are you taking prescribed painkillers incorrectly or selling them instead of using them as directed?
    \item Do you currently have legal issues that might result in incarceration, fines, or difficulties in renting housing?
    \item Do you engage in risky behaviors, such as exchanging sex for money, running drugs, having unprotected sex with strangers, sharing needles, or similar activities?
    \item Do you have a learning disability, developmental disability, or any other impairment?
    \item Do you have a mental health issue or concern?
    \item Do you have any mental health or cognitive issues that make it difficult to live independently?
    \item Do you have any physical disabilities that limit the type of housing you can access or make it difficult to live independently?
    \item Do you have planned activities—aside from mere survival—that make you feel happy and fulfilled?
    \item Do you receive income from the government, a pension, an inheritance, informal work, or a regular job?
    \item Do you suffer from any chronic health issues involving your liver, kidneys, stomach, lungs, or heart?
    \item Does anyone force or trick you into doing things against your will?
    \item For female respondents only: Are you currently pregnant?
    \item Has your alcohol or drug use resulted in you being kicked out of an apartment or shelter program in the past?
    \item Has your current period of homelessness been caused by experiencing emotional, physical, psychological, sexual, or other trauma? (Please answer YES or NO)
    \item Have you been attacked or beaten up since you became homeless?
    \item Have you been hospitalized as an inpatient?
    \item Have you ever had to leave your apartment, shelter program, or other living arrangement because of physical health issues?
    \item Have you experienced a head injury in the past?
    \item Have you received health care at an emergency department or room?
    \item Have you spent one or more nights in a holding cell, jail, or prison, regardless of the duration?
    \item Have you taken an ambulance to the hospital?
    \item Have you talked to the police because you witnessed a crime, were a victim or suspect, or were told to move along?
    \item Have you used a crisis service, such as those for sexual assault, mental health, family/intimate violence, distress, or suicide prevention?
    \item How long has it been since you lived in permanent, stable housing?
    \item If space were available in a program that specifically assists people living with HIV or AIDS, would you be interested?
    \item In the last three years, how many times have you experienced homelessness?
    \item In the last year, have you threatened or attempted to harm yourself or someone else?
    \item Is there any person or entity (for example, a past landlord, business, bookie, dealer, or government group like the IRS) that believes you owe them money?
    \item Is your current homelessness caused by a relationship breakdown, an unhealthy or abusive relationship, or actions by family or friends leading to eviction?
    \item When you are sick or not feeling well, do you avoid seeking help?
    \item Where do you sleep most frequently? (choose one)
    \begin{itemize}
        \item If you selected `Other', please specify further details.
    \end{itemize}
    \item Will alcohol or drug use make it difficult for you to maintain or afford housing?

\end{enumerate}

\subsection{VI-F-SPDAT}
\label{appens:vifsp}
\begin{enumerate}

    \item Has any child in the family experienced abuse or trauma in the last 180 days?
    \item Stayed one or more nights in a holding cell, jail, or prison, whether that was a short-term stay like the drunk tank, a longer stay for a more serious offense, or anything in between?
    \item Are any prescribed painkillers not taken as directed or being sold instead of used?
    \item Are there any children in your family aged 11 or younger?
    \item Are there any children in your family aged 6 or younger?
    \item Are there any family legal issues currently in or needing court resolution that might impact your housing or household composition?
    \item Do children aged 12 or younger spend 2 or more hours per day in activities?
    \item Do children aged 13 or older spend 3 or more hours per day in activities?
    \item Do you anticipate any additional adults or children joining your household within the first 180 days of being housed?
    \item Do your older children spend 2 or more hours per day helping their younger siblings with tasks such as preparing for school, homework, dinner, or bathing?
    \item Does any single member of your household have a medical condition, mental health concern, and issues with substance use?
    \item Does anyone in your family currently have legal issues that might lead to incarceration, fines, or difficulties in renting housing?
    \item Does anyone in your family force or trick you into doing things against your will?
    \item Does anyone in your family have a learning disability, developmental disability, or any other impairment?
    \item Does anyone in your family have a mental health issue or concern?
    \item Does anyone in your family have mental health or cognitive issues that make independent living challenging?
    \item Does anyone in your family have physical disabilities that limit the type of housing accessible or make independent living challenging?
    \item Does anyone in your family suffer from chronic health issues involving their liver, kidneys, stomach, lungs, or heart?
    \item Does every member of your family have planned activities—beyond just survival—that bring happiness and fulfillment?
    \item Does your family engage in risky behaviors such as exchanging sex for money, running drugs, having unprotected sex, sharing needles, or similar activities?
    \item Does your family have two or more planned activities each week, such as outings, library visits, family gatherings, or movie nights?
    \item Does your family receive income from the government, a pension, inheritance, informal work, or a regular job?
    \item Has alcohol or drug use by anyone in your family resulted in being evicted from your housing or shelter program?
    \item Has any family member been hospitalized as an inpatient?
    \item Has any family member received care at an emergency department or room?
    \item Has any family member taken an ambulance to the hospital?
    \item Has anyone in your family been attacked or beaten up since becoming homeless?
    \item Has anyone in your family experienced a head injury in the past?
    \item Has your family composition changed in the last 180 days due to factors such as divorce, children rejoining, someone leaving for military service or incarceration, or a relative moving in?
    \item Has your family ever had to leave an apartment, shelter, or other accommodation due to physical health issues?
    \item Has your family spoken to the police because they witnessed a crime, were victims, or were suspects, or because the police advised them to move along?
    \item Has your family used a crisis service such as those for sexual assault, mental health, family/intimate violence, distress, or suicide prevention?
    \item Has your family's current homelessness been caused by experiencing emotional, physical, psychological, sexual, or other trauma? (Please answer YES or NO)
    \item Have any children been removed from your family by a child protection service in the last 180 days?
    \item How long has it been since you and your family lived in permanent, stable housing?
    \item How many children under the age of 18 are currently with you?
    \item How many children under the age of 18 are not currently with your family but are expected to join you when you get housed?
    \item How many parents are in your family?
    \item If space were available in a program for people living with HIV or AIDS, would your family be interested?
    \item If there are school-aged children, do they attend school most weeks?
    \item If your household includes a female, is any member currently pregnant?
    \item In the last 180 days, have any children lived with family or friends due to your housing situation or homelessness?
    \item In the last three years, how many times has your family experienced homelessness?
    \item In the last year, has anyone in your family threatened or attempted to harm themselves or others?
    \item Is every member of your family capable of taking care of basic needs such as bathing, changing clothes, using the restroom, obtaining food, and accessing clean water?
    \item Is the head of your household 60 years of age or older?
    \item Is there any medication that a doctor prescribed for you or your family that is not being taken?
    \item Is there any person or entity (e.g., past landlord, business, bookie, dealer, government group like the IRS) that believes your family owes them money?
    \item Is your family's current homelessness caused by a relationship breakdown, an unhealthy or abusive relationship, or interventions by family or friends?
    \item Please provide a list of children's names and ages.
    \item When a family member is sick or unwell, does your family avoid seeking medical help?
    \item Where do you and your family sleep most frequently? (choose one)
    \begin{itemize}
        \item If you selected `Other', please specify further details.
    \end{itemize}
    \item Will alcohol or drug use make it difficult for your family to maintain or afford housing?
    
\end{enumerate}

\subsection{TAY-VI-SPDAT}
\label{appen:tay}
\begin{enumerate}

    \item Are you currently able to care for your basic needs (such as bathing, changing clothes, using a restroom, obtaining food, and accessing clean water)?
    \item Are you not taking any medications that a doctor prescribed for you?
    \item Are you taking prescribed painkillers incorrectly or selling them instead of using them as directed?
    \item Did conflicts regarding gender identity or sexual orientation contribute to your homelessness?
    \item Did differences in religious or cultural beliefs with your parents, guardians, or caregivers lead to your homelessness?
    \item Did you become homeless because you ran away from your family home, a group home, or a foster home?
    \item Do you currently have legal issues that might lead to incarceration, fines, or difficulties renting housing?
    \item Do you engage in risky behaviors (such as exchanging sex for money, food, drugs, or a place to stay; running drugs; having unprotected sex with strangers; sharing needles; etc.)?
    \item Do you have a learning disability, developmental disability, or any other impairment?
    \item Do you have a mental health issue or concern?
    \item Do you have any mental health or cognitive issues that hinder your ability to live independently?
    \item Do you have physical disabilities that limit your housing options or make living independently difficult?
    \item Do you have planned activities—beyond mere survival—that make you feel happy and fulfilled?
    \item Do you receive income from the government, an inheritance, an allowance, informal work, or a regular job?
    \item Do you suffer from chronic health issues involving your liver, kidneys, stomach, lungs, or heart?
    \item Does anyone force or trick you into doing things against your will?
    \item Has your alcohol or drug use resulted in eviction from your housing or shelter program?
    \item Have you been attacked or beaten up since you became homeless?
    \item Have you been hospitalized as an inpatient?
    \item Have you ever been pregnant, or have you ever gotten someone pregnant?
    \item Have you ever had to leave your living situation due to physical health issues?
    \item Have you experienced a head injury in the past?
    \item Have you received care at an emergency department or room?
    \item Have you spent one or more nights in a holding cell, jail, prison, or juvenile detention, regardless of the duration?
    \item Have you taken an ambulance to the hospital?
    \item Have you talked to the police because you witnessed a crime, were a victim or suspect, or were advised to move along?
    \item Have you used a crisis service (for sexual assault, mental health, family/intimate violence, distress, or suicide prevention)?
    \item How long has it been since you lived in permanent, stable housing?
    \item If space were available in a program for people living with HIV or AIDS, would you be interested?
    \item If you have tried marijuana, did you first try it at age 12 or younger?
    \item In the last three years, how many times have you experienced homelessness?
    \item In the last year, have you threatened or attempted to harm yourself or others?
    \item Is there any person or entity (e.g., a past landlord, business, bookie, dealer, or the IRS) that believes you owe money?
    \item Was your homelessness a result of violence among family members at home?
    \item Was your homelessness caused by an unhealthy or abusive relationship?
    \item Was your homelessness caused by your family or friends?
    \item Were you ever incarcerated before the age of 18?
    \item When you are sick or unwell, do you avoid seeking medical help?
    \item Where do you sleep most frequently? (choose one)
    \begin{itemize}
        \item If you selected `Other', please specify further details.
    \end{itemize}
    \item Will alcohol or drug use make it difficult for you to maintain or afford housing?
\end{enumerate}

\end{document}